\documentclass[10pt]{article}
\usepackage{amssymb,latexsym,amsmath,amscd,amsthm,natbib}
\usepackage{amsmath,amssymb} 
\usepackage{amsthm} 
\usepackage{hyperref}
\usepackage{graphicx}
\usepackage{comment}
\usepackage{subcaption}
\usepackage{epsfig}
\usepackage[active]{srcltx}
\usepackage{psfrag}
\usepackage{enumerate}
\usepackage{color}
\theoremstyle{defi}
\newtheorem{defi}{Definition}[section]
\newtheorem{ex}[defi]{Example}
\theoremstyle{plain}
\newtheorem{lema}[defi]{Lemma}
\newtheorem*{lema*}{Lemma}

\newtheorem{thm}[defi]{Theorem}

\newtheorem{prop}[defi]{Proposition}
\newtheorem{cor}[defi]{Corollary}
\newtheorem{rk}[defi]{Remark}

\definecolor{pink}{rgb}{1,0,1}

\newtheorem{teo-def}[defi]{Theorem/Definition}

\newcommand{\ZZ}{\mathbb{Z}}
\newcommand{\CC}{\mathbb{C}}

\newcommand{\RR}{\mathbb{R}}

\newcommand{\dd}{\mathrm{diag}}

\newcommand{\Arg}[1]{\theta_{#1}}
\newcommand{\Log}{\mathrm{Log}}

\newenvironment{proofThm11}{\noindent {\textit{Proof of Theorem 1.1.}}}{$\square$ \vspace{3mm}}

\usepackage{anysize}

\marginsize{2.5cm}{2.5cm}{1.5cm}{3cm}

\newif\ifprivate
 \privatetrue

\def\???{\ifprivate {\bf {???}} \marginpar{{\Huge {\bf ?}}}
\else \fi}

 \definecolor{zielony}{rgb}{0.5, 0.9, 0.1}
 \definecolor{czerwony}{rgb}{0.9, 0.2, 0.1}
 \definecolor{niebieski}{rgb}{0.3, 0.1, 0.9}

\definecolor{garnet}{RGB}{210,15,30}

\title{Embeddability and rate identifiability of \\ Kimura 2-parameter matrices}
\author{Marta Casanellas, Jes\'us Fern\'andez-S\'anchez, Jordi Roca-Lacostena}

\begin{document}

\maketitle

\begin{abstract}

Deciding whether a {substitution} matrix is embeddable (i.e. the corresponding Markov process has a continuous-time realization) is an open problem even for $4\times 4$ matrices. We study the embedding problem and rate identifiability for the K80 model of nucleotide substitution. For these $4\times 4$ matrices, we fully characterize the set of embeddable K80 Markov matrices  and the set of embeddable matrices for which rates are identifiable. In particular, we describe an open subset of embeddable matrices with non-identifiable rates. This set contains matrices with positive eigenvalues and also diagonal largest in column matrices, which might lead to consequences in parameter estimation in phylogenetics. Finally, we compute the relative volumes of embeddable K80 matrices and of embeddable matrices with identifiable rates. This study concludes the embedding problem for the more general model K81 and its submodels, which had been initiated by the last two authors in a separate work.

\emph{Keywords}. Nucleotide substitution model; Markov matrix; Markov generator; matrix logarithm; embedding problem; rate identifiability.
\end{abstract}

\section{Introduction}

Modeling {molecular substitutions} is the first and arguably the most fundamental step in phylogenetics and it is where the crucial hypotheses that should allow us to reconstruct the evolutionary history are to be assumed. There are several approaches one might take and it is natural to ask which conditions must be required for evolutionary models to fit the real evolutionary processes. For instance, under a Markov process, which are the plausible structures for {substitution} matrices in a given nucleotide or amino acid substitution model? Should any Markov matrix with such structure be considered as a biologically realistic {substitution} matrix?
Concerned with these and similar questions, we address the study of the embedding problem for the matrices of nucleotide substitution models. The final aim of this problem is to characterize those Markov matrices that are consistent with a homogeneous continuous-time approach of evolution. In other words, the aim is to decide whether a given Markov matrix in the model can be written as the exponential of some rate matrix. In this case, the Markov matrix is said to be \emph{embeddable}  (or to have a \emph{continuous realization} as in the book by \citealt{steelbook}), and the rate matrix is called a \emph{Markov generator} of the Markov matrix.

The motivation and applications of the embedding problem are diverse and include economics, social sciences, and molecular evolution  \cite[see][for example]{Israel, Singer,jia2,ChenJia,verbyla}. The problem has been addressed in a number of papers \citep[e.g.][]{Culver,Cuthbert72,Cuthbert73,Davies,goodman1970,Guerry2013,vanbrunt,Guerry2019}, {and it has been solved} for $2\times2$ and $3\times 3$ matrices by \citet{Israel,goodman1970,Guerry2019}. {However, }it is far from being solved for larger matrices with full generality. Because our motivation and interest arises from the study of mathematical properties of nucleotide substitution models, we focus on the particular case of $4\times 4$ matrices and, more specifically, on the Kimura models of nucleotide substitution \emph{K80} and \emph{K81}. These models were proposed by \citet{kimura80,kimura81} in the celebrated papers published almost 40 years ago. The main motivation was to include different parameters for different nucleotide substitution types, which according to biological data are not equally likely.
The {substitution} matrices of nucleotide replacements for the K81 model are Markov matrices with the following structure
$$\begin{small}
\begin{pmatrix}
 a	& b & c & d	\\
 b & a & d & c \\
 c & d & a	& b \\
 d & c & b & a
 \end{pmatrix},
\end{small}$$ {where the rows and columns are labelled by nucleotides adenine, guanine, cytosine and thymine (in this order) and the entry $(i,j)$ is the conditional probability that nucleotide $i$ is replaced by nucleotide $j$}. The K80 model is the submodel obtained when imposing that $c=d$. {In this case the parameter $c=d$ corresponds to transversions, which are substitutions from purines (adenine and guanine) to pyrimidines (cytosine and thymine) and viceversa, and parameter $b$ corresponds to transitions (substitutions within purines or within pyrimidines)}.
 When $b=c=d$, the resulting submodel is the Jukes-Cantor model JC69 introduced by \citet{JC69}.
%

In the phylogenetic setting, the embedding problem has been recently studied by \citet{ChenJia} and \citet{JJ}.
{In the first paper, the author addresses the embedding problem restricted to time-reversible Markov generators. }
 In the second, the authors characterize the embeddability of K81 matrices with different eigenvalues.
In a parallel work,  { \cite{Kaie} characterize} those embeddable matrices in a group-based model whose Markov generators satisfy the constraints of the model; this is usually referred to as \emph{model embeddability}.
%
{
Although the problem of model embeddability is natural from a continuous-time approach, in this paper we {do not impose time-reversibility or any model restriction on rates, as we} are more interested in determining whether a Markov process has a homogeneous continuous-time realization or not (independently of the structure of the instantaneous rate matrix). On the other hand, \citet{kimura81} initiates his argument in terms of differences of nucleotides between two homologous sequences, so his model for transition and transversion parameters is firstly justified in terms of probabilities rather than instantaneous mutation rates. As shown here, rates do not necessarily satisfy the same symmetries as probabilities do (see Example 4.3).}

Embeddability has not been yet  fully characterized for K81 matrices with a negative eigenvalue of even multiplicity and, in particular, it has not been solved for matrices within the K80 model.
In the present work we address the study of the repeated eigenvalues case to fill this gap and give a complete answer for the K80 model.
Moreover, together with the results by \citet{JJ}, our results fully solve the embedding problem for the K81 model. More precisely, the results on the embeddability of K81 matrices can be summarized as follows (see also Remark \ref{rmk_embedk81}):
\begin{itemize}
\item[(1)] for {generic} K81 Markov matrices (where $b,c,d$ are different), embeddability holds if and only if the principal logarithm $Log(M)$ is a rate matrix;
\item[(2)] for K81 Markov matrices with exaclty two equal off-diagonal entries, embeddability is characterized by the theorem below (and its analogous versions under permutations of rows and columns, { i.e., if $b=d$ one has to exchange the roles of $b$ and $c$, and if $b=c$, then $b$ should be replaced by $d$.});
\item[(3)] for the JC69 model, embeddability holds if and only if the determinant is positive.
\end{itemize}

Moreover, in this paper, we go further in the study of the K80 {model} and derive a criterion to determine whether the mutation rates of embeddable matrices are \textit{identifiable} and, if not, we determine whether there is a finite number of Markov generators. {One says that the rates of a Markov matrix $M$ are \textit{identifiable} there is a unique Markov generator for $M$. The concept of non-identifiability of rates was called \textit{nonunique mapping} in the paper by \citet{Kaehler2015}.}

All together, the main results of this work for the K80 model are summarized in the following result:

\begin{thm}\label{thm:entriesIdentiff}
For any K80 Markov matrix
$M =\begin{footnotesize}
	\begin{pmatrix}
		a & b & c & c\\
		b & a & c & c\\
		c & c & a & b\\
		c & c & b & a
	\end{pmatrix}
\end{footnotesize}$ { with $b\neq c$}, the following holds:
\begin{enumerate}[(a)]
	\item If ${2}c=1-2b$, then $M$ is not embeddable.
	
	\item If ${2}c<1-2b$, $M$ is embeddable if and only if $c \leq \sqrt{b} -b$. In this case, {$Log(M)$ is a Markov generator and}
	\begin{enumerate}[i)]
		\item if $c < \frac{1}{4}-\frac{e^{-4\pi}}{4}$ then the rates of $M$ are identifiable,
		\item if $c =\frac{1}{4}-\frac{e^{-4\pi}}{4}$ then $M$ has exactly $3$ Markov generators,
		\item if $c > \frac{1}{4}-\frac{e^{-4\pi}}{4}$ then $M$ has infinitely many Markov generators.
	\end{enumerate}

	\item If ${2}c>1-2b$, $M$ is embeddable if and only if $\frac{1}{4} - \frac{e^{-2\pi}}{4} \leq c \leq \sqrt{b} -b $. In this case, the rates of $M$ are not identifiable and
	\begin{enumerate}[i)]
		\item if $c =\frac{1}{4}-\frac{e^{-2\pi}}{4}$ then $M$ has exactly $2$ Markov generators.
		\item if $c > \frac{1}{4}-\frac{e^{-2\pi}}{4}$ then $M$ has infinitely many Markov generators.
	\end{enumerate}
\end{enumerate}
\end{thm}

Case (b) above corresponds to the case of positive eigenvalues, for which we have that $M$ is embeddable if and only if its principal logarithm is a rate matrix (Corollary \ref{cor:CharOfEmb}). This result is in accordance with the well known case of different and real eigenvalues (see \citealt{Culver} and (1) above). Besides this,
the theorem also provides a full description of those embeddable K80 matrices with positive eigenvalues and non-identifiable rates \citep[see also][Example \ref{ex:pos_eig} and Section 3]{JJ}. Among these K80 matrices, there are some \emph{diagonal largest in column matrices}  (DLC), which are of relevance in parameter estimation in phylogenetics (see Remark \ref{DLC} and the Discussion).
As far as we are aware, this is the first result in this direction.
%


From the results obtained here, we recover the fact that among all possible Markov generators of an embeddable K80 matrix at most one keeps the K80 structure, which in turn coincides with the principal logarithm of the matrix \citep[see][]{JJ,ChenJia}.

Another goal of the paper is to quantify the difference between restricting to the homogeneous continuous-time models or considering any Markov matrix within the K80 model. In this direction, we compute the relative volume of embeddable K80 matrices inside the whole set of K80 Markov matrices. It turns out that embeddable matrices only account for about 35\% of all the {substitution} matrices. Similar computations exhibit that, although embeddable matrices with non-identifiable rates describe a set of positive measure within the K80 model (containing the embeddable matrices in case $(c)$ above), {99.99\% of embeddable K80 matrices have} identifiable rates, so the non-identifiability of rates is not representative of the general situation (see Table \ref{tab:volume} for the precise figures).

The organization of the paper is as follows. In Section \ref{sec:Preliminaries}
we introduce K80 matrices, state the embedding and rate-identifiability problem with precision and recall some known results that are relevant for posterior work. In Section \ref{sec:Embeddability} we characterize embeddability for K80 matrices both in terms of their eigenvalues (Corollary \ref{cor:VapsEmbed}) and their entries (Corollary \ref{cor:EntriesEmbed}), which together with the results already known for K81 matrices do fully solve the embedding problem for the K81 model. Later in Section \ref{sec:identifiability} we solve the rate identifiability problem by providing sufficient and necessary conditions for the rates to be identifiable (Proposition \ref{prop:identIFF}) and characterize rate identifiability in terms of the eigenvalues or the entries of the Markov matrix (Theorem \ref{thm:vapsIdentiff} and Corollary \ref{cor:entriesIdentiff2}). In Section \ref{sec:Volumes} we compute the relative volumes of embeddable and rate identifiable matrices within the K80 model and some biologically relevant subsets (see Table \ref{tab:volume}). Finally, in Section \ref{sec:Discussion} we discuss the implications and connections with other papers and possibilities for future work.

\section{Preliminaries}\label{sec:Preliminaries}

In this section we introduce the embedding problem for Markov matrices, {the nucleotide substitution model we work with, and a few known results related to the embedding problem needed for the sequel}.

\begin{defi} A matrix $M\in M_n(\RR)$ is said to be a \textit{Markov matrix} (or \emph{{substitution} matrix}) if its entries are positive or zero and its rows sum up to one. Similarly, $Q\in M_n(\RR)$ is said to be a \textit{rate matrix} if its rows sum to zero and { its off-diagonal entries are non-negative}.
\end{defi}

If $Q$ is a rate matrix, it is well-known that $e^{tQ}=\sum_{n\geq 0} \frac{t^n Q^n}{n!}$ is a Markov matrix for all $t\geq 0$, thus $Q$ is said to be a \emph{Markov generator} for $M=e^Q$ \citep{Davies}.
However, not every Markov matrix can be obtained this way. A Markov matrix $M$ is said to be \emph{embeddable} if $M=e^{Q}$ for some rate matrix $Q$. Characterizing which Markov matrices are embeddable is known as the \emph{embedding problem}. In terms of Markov processes, a Markov matrix is embeddable if it has a realization as a homogeneous continuous-time { Markov chain} ($M=e^{tQ}$ for some $t>0$ and some rate matrix $Q$).

We say that $Q\in M_n(\CC)$ is a \textit{logarithm} of a matrix $M$ if $e^Q=M$. 
As in the case of complex numbers, matrix logarithms are not unique \citep{Gantmacher}. Hence an embeddable matrix may admit more than one Markov generator. If an embeddable matrix has only one Markov generator we say that its \textit{rates are identifiable}. The \textit{identifiability problem} consists on deciding whether the rates are identifiable or not.

\begin{defi}\label{def:log}
Given $z\in\CC\setminus\{0\}$, let $\Arg z\in(-\pi,\pi]$ be its \textit{principal argument}.
For any given $k \in \ZZ$ we define the \textit{k-th determination of its logarithm} as $\log_k(z) = \log(|z|) + (\Arg z +2\pi k) i$. We denote by $\log(z)$ its \textit{principal logarithm} $\log_0(z)$. The \textit{principal logarithm} of an invertible matrix $M$, denoted as $\Log(M)$, is defined as the unique logarithm such that the imaginary part of its eigenvalues lies in the interval $(-\pi,\pi]$ \citep{Higham}. If $M$ diagonalizes, $M = P\;\dd(\lambda_1,\dots,\lambda_n)\;P^{-1}$, then its principal logarithm can be computed as $\Log(M)=P\;\dd(\log(\lambda_1),\dots,\log(\lambda_n))\;P^{-1}$.
\end{defi}

Our goal in this work is to deal with both the embedding problem and the identifiability problem for Kimura 2-substitution types model (K80 for short) \citep{kimura80} which is a submodel of Kimura 3-substitution types model K81 \citep{kimura81} and contains the most simple Jukes-Cantor model JC69 \citep{JC69}.

\begin{defi}
\label{def:K3Matrix}
A matrix $M\in M_4(\RR)$ is a \textit{K81 matrix} if it is of the form $$M=\begin{small}
\begin{pmatrix}
 a & b & c & d\\
 b & a & d & c \\
 c & d & a & b \\
 d & c & b & a \\
 \end{pmatrix}
\end{small}.$$  
For ease of reading, we will use the notation $M=K(a,b,c,d)$. A \textit{K80 matrix} is a K81 matrix with $c=d$. Similarly, a \textit{JC69 matrix} is a K81 matrix with $b=c=d$.\

When a K81 matrix is also a Markov matrix, we speak about \emph{K81 Markov matrices} (respectively \textit{K80 Markov matrices} and \textit{JC69 Markov matrices}). When the rows of a K81 matrix sum to zero and the off-diagonal entries are non-negative, then it is a \emph{rate K81 matrix} (and analogously for K80 and JC69).


{As a variation of the embedding problem  one may consider the problem of \emph{model embeddability} mentioned in the Introduction (which attempts to determine those Markov matrices in a certain model that have a Markov generator preserving the same identities that characterize the model). For instance, one may want to characterize those Markov matrices for which there is a K81 Markov generator.
For this model (and its submodels), model embeddability has been characterized in terms of inequalities in the eigenvalues of the Markov matrices \citep{JJ}. The model embeddability for general group based models has been studied by \citet{Kaie}. As we explained in the Introduction, we shall not focus in this problem in the present paper.}
%

\end{defi}

It is well known that all K81 matrices diagonalize under the following Hadamard matrix {\citep{EvansK3,Hendy1993}}: \begin{equation}\label{defS}
S: = \begin{footnotesize}
\begin{pmatrix}
1 & 1 & 1 & 1\\
1 & 1 & -1 & -1\\
1 & -1 & 1 & -1\\
1 & -1 & -1 & 1\\
\end{pmatrix}
\end{footnotesize}.
\end{equation}
Furthermore, it follows from a straightforward computation that a K81 matrix has the following eigenvalues: \begin{equation} \label{eq:eigenvalues}
\begin{matrix}
a+b+c+d &\, & x:=a+b-c-d  \\
 y:=a-b+c-d& \, & z:=a-b-c+d.
\end{matrix}
\end{equation}
In particular, any K81 matrix has real eigenvalues. If we deal with Markov matrices, then the eigenvalues become $1=a+b+c+d$, $x=1 -2c-2d$, $y= 1-2b-2d$, and $z=1-2b-2c$.
For a K80 Markov matrix the eigenvalues become $1$, $x=1-4c$, and {$y=z=1-2b-2c$} . Similarly, the eigenvalues of a JC69 Markov matrix are $1$ and {$x=y=z=1-4b$}.
{Note that this allows us to parametrize all Markov K81 matrices by the entries $b,c,d$ ($a=1-b-c-d$) or by their eigenvalues $x,y,z$. In particular, for the K80 model, the bijection between both spaces of parameters is given by:}
\begin{equation}\label{eq:paramBijection}
	\begin{matrix}
		\varphi: & \text{Entries}  & \longrightarrow & \text{Eigenvalues}\\
    		  & (b,c) & \longmapsto   & (1-4c \ , \ 1-2b-2c)\\
	\end{matrix} \hspace*{7mm}
	\begin{matrix}
		\varphi^{-1}: & \text{Eigenvalues} & \longrightarrow & \text{Entries}\\
    				& (x,y)   & \longmapsto	 & (\frac{1+x-2y}{4} \ , \ \frac{1-x}{4})\\
	\end{matrix}
\end{equation}

\begin{rk}\label{rmk_vaps}\rm It is worth noting that if $M$ is a K80 Markov matrix with eigenvalues {$1,x,y,y$} then
\begin{enumerate}[a)]
	\item $|x|\leq 1$, $|y|\leq 1$ (by Perron-Frobenius theorem).
  	\item If $y\neq x$, then $y\neq 1$. Indeed, the equality (\ref{eq:eigenvalues}) shows that the eigenvalues of $M=K(1-b-2c,b,c,c)$ are $x= 1-4{c}$ and $y=1-2b-2c $. Since $b,c \geq 0$, if $y=1$ it follows that {$b=c=0$} and $x=1$, which contradicts $y \neq x$. Therefore we have $y\neq 1$.
 	\item If $y\neq x$ and $x<0$, then $M$ does not have any real logarithm \citep[see][]{Culver}.
 	\item If $x$ or $y$ are zero, then $M$ is not embeddable because it has zero determinant ({note that $\det e^Q=e^{\mathrm{tr}(Q)}> 0$}).
 	\item {If $x=y$ then $M$ is a JC69 matrix. In this case, $M$ is embeddable if and only if $x>0$ \citep{JJ}}.
\end{enumerate}
\end{rk}

It is known that if a matrix has determinant close to $1$ then the principal logarithm is the only possible real logarithm of that matrix \citep{Israel,Cuthbert72,Cuthbert73,Singer}. The same holds for matrices with distinct real eigenvalues \citep{Culver}. The next theorem provides sufficient and necessary conditions for the principal logarithm of a K81 Markov matrix to be a rate matrix (\citealp{JJ}, see also example 4.5 by \citealt{Kaie}). In particular, the result below solves the embedding and identifiability problems in the K81 model except for those matrices with repeated eigenvalues. {At the same time, this result solves the model embeddability problem for the K81 model and its submodels.}

\begin{thm}[\citealt{JJ}, Corollary 3.5]\label{thm:LogK3Embed}
Let $M$ be a K81 Markov matrix with eigenvalues $1,x,y,z$.
Then,
	\begin{enumerate}[i)]
		\item $\Log(M) = S \; \dd(0,\log(x),\log(y),\log(z)) \; S^{-1}$ is a K81 matrix. Furthermore, it is the only logarithm of $M$ that is itself a K81 matrix.

 		\item $\Log(M)$ is a rate matrix if and only if
 			\begin{eqnarray}
				x,y,z >0, \qquad \label{ineq}
				x\geq y z, \qquad
				y\geq x z, \qquad
				z\geq x y.
 			\end{eqnarray}
	\end{enumerate}
\end{thm}


If a Markov matrix has a repeated eigenvalue (e.g. K80 and JC69 matrices), then it {may have} infinitely many real logarithms \citep{Culver} and hence we must check if any of them is a rate matrix before deciding that such a matrix is not embeddable. Nonetheless, it is known that any logarithm (including non-real logarithms) of a given matrix $M$ can be obtained as $\Log(M) + L$ where $L$ is one out of infinite logarithms of the identity that commute with $M$ \citep{Higham}.\\

{The following notation will be used throughout the paper.}
{
\emph{Note.}  
$Id_n$ denotes the identity matrix of order $n$. We write $GL_n(\mathbb{K})$ for the space of $n\times n$ invertible matrices with entries in $\mathbb{K}= \RR$ or $\CC$. Given a matrix $M$, we define the \emph{commutant} of $M$, $Comm^*(M)$, as the set of invertible complex matrices that commute with $M$.

}
\begin{rk}\label{rk_comm}
\rm {If $D$ is a diagonal matrix,  $D=\dd(\overbrace{\lambda_1,\dots,\lambda_1}^{m_1},\overbrace{\lambda_2,\dots,\lambda_2}^{m_2},\dots,\overbrace{\lambda_n,\dots, \lambda_n}^{m_n})$, then $Comm^*(D)$ consists on all the block-diagonal matrices whose blocks are taken from the corresponding $GL_{m_i}(\CC)$. In particular, if all $\lambda_i$ are different, $Comm^*(D)$ is the set of invertible diagonal matrices.
%
%
Moreover, if $M=P M' P^{-1}$, then $Comm^*(M)$ is formed by all matrices that arise as $P\, U\, P^{-1}$ with $U\in Comm^*(M')$.
}
\end{rk}

Using that all K80 matrices can be diagonalized by $S$,  {immediate application of Theorem 1.27 by \citet{Higham} allows us to enumerate all the logarithms (real or not) of any K80 Markov matrix:}


\begin{thm}\label{Thm:EnumLogK2} Given a K80 Markov matrix $M$ with eigenvalues $1$, $x$, $y$, $y$ all (complex) solutions to $exp(Q) = M$ are given by:
 \begin{equation}\label{eq:HighEnum}
	Q= S\; U \; \dd\big(\log_{k_1}(1),\log_{k_2}(x),\log_{k_3}(y),\log_{k_4}(y) \big) \; U^{-1} \; S^{-1}
 \end{equation}
where $k_i \in \ZZ$, $S$ is defined in \eqref{defS} and $U \in Comm^*\big(\dd(1,x,y,y)\big)$.
\end{thm}

{
%
The previous theorem provides the eigendecomposition of any complex logarithm of a K80 Markov $M$ while making explicit the connection with the commutant of $M$ (see Remark \ref{rk_comm}).
%
In the next section, we will consider only \emph{real} logarithms, so we will need to pay special attention to the corresponding restrictions on their eigenvalues and  eigenvectors.
}

\section{Embeddability of K80 Markov matrices}\label{sec:Embeddability}



In this section we characterize all the real logarithms with rows summing to zero of K80 Markov matrices and, as a consequence, we are able to provide sufficient and necessary conditions for such a matrix to be embeddable. {These} conditions are given in terms of some inequalities involving either the eigenvalues of the matrix (Corollary \ref{cor:VapsEmbed}) or its entries (Corollary \ref{cor:EntriesEmbed}). 

Theorem \ref{Thm:EnumLogK2} allows us to compute all the logarithms of any given K80 matrix by using the principal logarithm and logarithms of the identity. The main issue with the description given by this theorem is that we are interested only in logarithms that satisfy rate matrices constraints. The forthcoming Proposition \ref{prop:L1L2} {characterizes all the real logarithms of a K81 matrix whose rows sum to 0. To this end we introduce the following matrices:}

\begin{defi}\label{def:L1L2}
Let $M$ be a K80 Markov matrix with eigenvalues $1$, $x$, $y$, $y$ satisfying $x >0$, {$y \neq x$ and $y\neq0$}. Given $k\in \ZZ$ and $A\in GL_2(\RR)$ we introduce the following notation:
	\begin{itemize}
		\item $L_0 := S\;\dd \big(0,\log(x),\log|y|,\log|y| \big) \; S^{-1}$.
		\item $L_1(A) := \big( S\; \dd(Id_2,A) \big) \; \dd \left(\begin{footnotesize}
			\begin{pmatrix}
				0 & 0\\
				0 & 0\\
			\end{pmatrix}\end{footnotesize},
			\begin{footnotesize} \begin{pmatrix}
				0 & 1\\
				-1 & 0\\
			\end{pmatrix}\end{footnotesize}
			\right) \big( \dd\big(Id_2,A^{-1}\big) \; S^{-1} \big)$.

		\item $Q(k,A) := L_0 + \big(2\pi k +\Arg y \big) L_1(A)$.
	\end{itemize}
\end{defi}

{Although} these matrices depend on $M$, for ease of reading we decided not to reflect it in the notation. {Note that $\Arg{y}$ is either $0$ or $\pi$ because $y$ is a real number (see (\ref{eq:eigenvalues})). Also note that if $y>0$, then $L_0=Log(M)$ and $Q(0,A)=Log(M)$ for all $A\in GL_2(\RR)$.}


\begin{prop}\label{prop:L1L2}
Let $M$ be an invertible K80 Markov matrix with eigenvalues $1$, $x$, $y$, $y$ satisfying $x >0 $ and $y \neq x$. Then for any matrix $Q\in M_4(\RR)$ the following are equivalent:
	\begin{enumerate}[i)]
		\item $Q$ is a real logarithm of $M$ whose rows sum to 0.
		\item $Q = Q(k,A)$ for some $k\in \ZZ$ and $A\in GL_2(\RR)$.
	\end{enumerate}
\end{prop}

It is worth pointing out that the previous definition and proposition can be generalized to $4\times 4$ diagonalizable matrices. We do not state them in this full generality in order to make the notation and the proof more readable. {The proof of this result relies on  Theorem \ref{Thm:EnumLogK2} and the study of the commutant of K80 matrices}.

\begin{proof} Note that Remark \ref{rmk_vaps} shows that $y\neq 1$.

\noindent
$i) \Rightarrow ii)$  Assume that $Q$ is a real logarithm of $M$ whose rows sum to 0. According to Theorem \ref{Thm:EnumLogK2} it holds that
\begin{equation}\label{eq:AllLogs}
		Q = S\; U \; \dd\big(\log_{k_1}(1),\log_{k_2}(x),\log_{k_3}(y),\log_{k_4}(y) \big) \; U^{-1} \; S^{-1}
\end{equation}
for some $k_i \in \ZZ$ and some $U \in GL_4(\CC)$ that commutes with $\dd\big(1,x,y,y\big)$.

{Since the rows of $Q$ sum to $0$, we have that $(1,1,1,1)$ is an eigenvector of $Q$ with eigenvalue $0$. Moreover, since the matrix $Q$ is real, complex non-real eigenvalues of $Q$ come in conjugate pairs. Therefore, as $y\neq x$ and $y\neq 1$, $\log_{k_1}(1)$, $\log_{k_2}(x)$ are real and at least one of them is equal to $0$. In particular, $k_1=k_2=0$.}
%
%


Noting that $y\neq 1$, $y\neq x$ and using {Remark \ref{rk_comm}} we get:
\begin{center}
 $U=\begin{pmatrix}
		U_1 & 0\\
		0 & U_2\\
	\end{pmatrix}$
with $U_1 \in Comm^*\big( \dd(1,x)\big)$ and $U_2 \in GL_2(\CC)$. \end{center}

Now, as $U_1$ commutes with $\dd(1,x)$, it does also commute with $\dd\big(\log_0(1),\log_0(x)\big)=\dd\big(0,\log(x)\big)$ and hence:

\begin{eqnarray}
			S^{-1}\; Q \; S & = &  U \;
			\begin{small} \begin{pmatrix}
				\dd\big(0,\log(x)\big) & 0\\
				0 & \dd\big(\log_{k_3}(y),\log_{k_4}(y)\big)\\
			\end{pmatrix} \end{small} \;
			U^{-1} \nonumber\\
			 & = &\begin{small} \begin{pmatrix}
				U_1 \; \dd\big(0,\log(x) \big) \; U_1^{-1}& 0\\
				0 & U_2 \; \dd\big(\log_{k_3}(y), \log_{k_4}(y) \big) \; U_2^{-1}\\
			\end{pmatrix} \end{small}\nonumber\\
			 & = &\begin{small} \begin{pmatrix}
				\; \dd\big(0,\log(x) \big) \; & 0\\
				0 & U_2 \; \dd\big(\log_{k_3}(y),\log_{k_4}(y) \big) \; U_2^{-1}\\
			\end{pmatrix} \end{small}
			\label{eq_U}
			\, .
	\end{eqnarray}

{Since non-real eigenvalues and eigenvectors of $Q$ appear in complex conjugate pairs}, either {$\log_{k_3}(y)$, $\log_{k_4}(y)\in \RR$ or} $\log_{k_4}(y)= \overline{\log_{k_3}(y)}$ and the vector columns of $U_2$, namely $v$ and $w$, must be complex vectors satisfying $w = \lambda \overline{v}$ for some $\lambda \in \CC\setminus \{0\}$.
{In the first case we have that $k_3=k_4=0$ and $y>0$ so that $U_2$ commutes with $\dd\big(\log_{0}(y),\log_{0}(y) \big)$. Note that in this case we have $Q=Log(M)$ and coincides with $Q(0,A)$ for any $A$. For the second case we get:}
\begin{center}
	$U_2 = \big(v\; w \big) =
	\begin{pmatrix}
		\alpha+\beta i & \alpha-\beta i \\
		\gamma+\delta i & \gamma - \delta i\\
	\end{pmatrix} \; \dd(1,\lambda)$
\end{center}
for some $\lambda\in \CC\setminus \{0\}$, $\alpha,\beta,\gamma,\delta\in \RR$ such that $\det(U_2)= {2}\lambda(\beta\gamma-\alpha\delta)i \neq 0$.

Since the matrix $\dd(1,\lambda)$ commutes with any diagonal matrix {(and in particular with $\dd(\log_{k_3}(y),\overline{\log_{k_3}(y)})$, we obtain that \eqref{eq_U} is equivalent to the following equality:
$$S^{-1}QS=\begin{footnotesize} \begin{pmatrix}
			1 & 0 & 0 &0\\
			0 & 1 & 0 &0\\
			0 & 0 & \alpha+\beta i & \alpha-\beta i \\
			0 & 0 & \gamma+\delta i & \gamma - \delta i\\
	\end{pmatrix}\end{footnotesize} \begin{pmatrix}
				\; \dd\big(0,\log(x) \big) \; & 0\\
				0 &  \dd\big(\log_{k_3}(y),\overline{\log_{k_3}(y)} \big) \\
			\end{pmatrix}
			\begin{footnotesize}
			\begin{pmatrix}
			1 & 0 & 0 &0\\
			0 & 1 & 0 &0\\
			0 & 0 & \alpha+\beta i & \alpha-\beta i \\
			0 & 0 & \gamma+\delta i & \gamma - \delta i\\
	\end{pmatrix}^{-1}
	\end{footnotesize}.
	$$
	}
	
	Furthermore, by considering the following invertible matrix
\begin{equation}\label{eq:R}
	R = \begin{footnotesize} \begin{pmatrix}
		1 & 0 & 0 & 0\\
		0 & 1 & 0 & 0\\
		0 & 0 & 1 & 1\\
		0 & 0 & i &-i\\
	\end{pmatrix} \end{footnotesize} ,
\end{equation} we have
\begin{center}
	\begin{tabular}{r l}
		$ Q $&$= S\; \big(U\; R^{-1} \big) \big( \; R \; \dd\big(0,\log(x),\log_{k_3}(y),\overline{\log_{k_3}(y)} \big) \; R^{-1} \big) \big( R \; U^{-1} \big) S^{-1} $ \vspace{3mm}\\
 		&$=S\; \begin{footnotesize} \begin{pmatrix}
			1 & 0 & 0 & 0\\
			0 & 1 & 0 & 0\\
			0 & 0 & \alpha & \beta\\
			0 & 0 & \gamma &\delta\\
		\end{pmatrix} \; \begin{pmatrix}
			0 & 0 & 0 & 0\\
			0 & \log(x) & 0 & 0\\
			0 & 0 & \log|y|& 2\pi k_3 + \Arg y\\
			0 & 0 & -\big(2\pi k_3 + \Arg y \big) & \log|y|\\
		\end{pmatrix} \;\begin{pmatrix}
			1 & 0 & 0 & 0\\
			0 & 1 & 0 & 0\\
			0 & 0 & \alpha & \beta\\
			0 & 0 & \gamma &\delta\\
		\end{pmatrix}^{-1}\end{footnotesize}
		\; S^{-1} .$\\
	\end{tabular}
\end{center}
Using the distributive property and the fact that the invertible matrix
$A:= \begin{pmatrix}
	\alpha & \beta\\
	\gamma & \delta\\
\end{pmatrix}\in GL_2(\RR)$
commutes with $\dd(\log|y|,\log|y|)$ we can rewrite $Q$ as:
\begin{center}
	\begin{scriptsize}
		$S \; \begin{pmatrix}
			0 & 0 & 0 & 0\\
			0 & \log(x) & 0 & 0\\
			0 & 0 & \log|y|& 0\\
			0 & 0 & 0 & \log|y|\\
		\end{pmatrix} \;S^{-1} + \big(2\pi k_3 + \Arg y \big) \;S\; \begin{pmatrix}
			1 & 0 & 0 & 0\\
			0 & 1 & 0 & 0\\
			0 & 0 & \alpha & \beta\\
			0 & 0 & \gamma &\delta\\
		\end{pmatrix} \; \begin{pmatrix}
			0 & 0 & 0 & 0\\
			0 & 0 & 0 & 0\\
			0 & 0 & 0& 1\\
			0 & 0 & -1 & 0\\
		\end{pmatrix} \;\begin{pmatrix}
			1 & 0 & 0 & 0\\
			0 & 1 & 0 & 0\\
			0 & 0 & \alpha & \beta\\
			0 & 0 & \gamma &\delta\\
		\end{pmatrix}^{-1} \; S^{-1} .$
	\end{scriptsize}
\end{center}
That is, $Q= L_0+\big( 2\pi k_3 +\Arg y \big) L_1(A)$ and this concludes this part of the proof.\\

\noindent
$ ii) \Rightarrow i)$
By definition $Q(k,A)$ is a real matrix. Furthermore, $(1,1,1,1)$ is an eigenvector with eigenvalue $0$ for both $L_0$ and $L_1$ thus the rows of $Q(k,A)$ sum to $0$. A straightforward computation shows that $exp(Q(k,A))=M$ for any $A \in GL_2(\RR)$ and $k\in \ZZ$. {Indeed, note that the matrix $R$ defined in \eqref{eq:R} and the matrix $\dd(Id_2,A)$ both commute with $\dd(a,b,c,c)$ for any $a,b,c\in\RR$ so, the matrix defined as $U:= \dd(Id_2,A)\; R$ commutes with $\dd(0,\log(x),\log|y|,\log|y|)$. In particular we can write $L_0$ as}
	$$L_0 = S\; U \; \dd \big(0,\log(x),\log|y|,\log|y| \big) \; U^{-1} \; S^{-1} \, .$$
{On the other hand, we have}
	$$L_1(A) = S\; (U\; R^{-1}) \; \dd \left(
	\begin{footnotesize}\begin{pmatrix}
		0 & 0\\
		0 & 0\\
	\end{pmatrix} \end{footnotesize} ,
	\begin{footnotesize} \begin{pmatrix}
		0 & 1\\
		-1 & 0\\
	\end{pmatrix} \end{footnotesize} \right)
	(U\; R^{-1})^{{-1}} \; S^{-1} = (S\;U)\;\dd(0,0,i,-i) (U^{-1} \; S^{-1}).$$
Hence,
\begin{equation}\label{eq_VEPs}
Q(k,A) = (S\;U)\;\dd(\log(1),\log(x), \log_k(y),\log_l(y)) \;(U^{-1} \; S^{-1})
\end{equation}
with $l=-k$ if $y>0$ and $l=-k-1$ if $y<0$. {As $U$ also commutes with $\dd(1,x,y,y)$, the matrix $Q(k,A)$ is one of the matrices listed in Theorem \ref{Thm:EnumLogK2} and so, it is a logarithm of $M$.}\\
\end{proof}

By using the proposition above we obtain a parametrization of all real logarithms with rows summing to 0 of any K80 Markov matrix. Nonetheless, this parametrization is not injective because distinct choices of $A$ and $k$ may produce the same logarithm of $M$. Theorem \ref{Thm:QAlphaBeta} below provides an \textit{injective} parametrization for these logarithms (other than $\Log(M)$) by considering only matrices of the following form:
\begin{equation} \label{eq:QabDef}
	Q(k,\alpha,\beta):= Q(k,A) \text{ with } k\in \ZZ \text{ and }
	A=\begin{pmatrix}
  	1 & 0\\
   	\alpha & \beta\\
  \end{pmatrix} \text{ for some } \alpha \in \RR, \ \beta \in {\RR_{>0}}.
\end{equation}

\begin{rk}\label{rk:AtoAlphaBeta}
\rm {We claim that for any $A=(a_{ij}) \in GL_2(\RR)$ and any $k \in \ZZ$, there exist $\alpha,\beta,\tilde{k}$ such that $Q(k,A)=Q(\tilde{k},\alpha,\beta)$.  Indeed, one can easily check that $L_1(A)$ can be realized as }
\begin{center}
	{
	$L_1(A)=L_1 \left( \begin{pmatrix}
		1 & 0 \\
		\frac{a_{11} a_{21} + a_{12} a_{22}}{a_{11}^2 + a_{12}^2}& \frac{\det(A)}{a_{11}^2 + a_{12}^2}\\
	\end{pmatrix} \right)$ and $L_1(A)=- L_1 \left( \begin{pmatrix}
		1 & 0\\
		\frac{a_{11} a_{21} + a_{12} a_{22}}{a_{11}^2 + a_{12}^2}& - \frac{\det(A)}{a_{11}^2 + a_{12}^2}\\
	\end{pmatrix}\right)$.
	}
\end{center}
{Since we want the parameter $\beta$ to be positive, we will  consider one expression or the other depending on the sign of $\det(A)$. Since $L_0$ does not depend on $A$, we get that}:
	$$Q(k,A)=
	\begin{cases}
		Q\left(k,\frac{{a_{11} a_{21} + a_{12} a_{22}}}{a_{11}^2 + a_{12}^2}, \frac{\det(A)}{a_{11}^2 + a_{12}^2} \right) & \text{ if } \det(A)>0, \vspace*{5mm} \\
		Q\left(-k - \frac{\Arg y }{\pi},\frac{{a_{11} a_{21} + a_{12} a_{22}}}{a_{11}^2 + a_{12}^2}, - \frac{\det(A)}{a_{11}^2 + a_{12}^2}\right) & \text{ if } \det(A)<0 \, .\\
	\end{cases}$$
{Note that $\frac{\Arg y }{\pi}$ is either 0 (if $y>0$) or 1 (if $y<0$)}. Summing up,  by Proposition \ref{prop:L1L2}, any real logarithm $Q$ of $M$ with rows summing to $0$ can be expressed as $Q=Q(k,\alpha,\beta)$ for some $\alpha,\beta$.
\end{rk}

\begin{rk}\label{rk:VEPs}\rm {Note that the equality \eqref{eq_VEPs} obtained in the proof of Proposition \ref{prop:L1L2} allows the computation of eigenvalues and eigenvectors for $Q(k,\alpha, \beta)$. Indeed, \eqref{eq_VEPs} gives us that the eigenvalues of $Q(k,A)$ are $\log(1),\log(x), \log_k(y),\overline{\log_k(y)}$ and the eigenvectors  can be chosen as the columns of the product of matrices
$$S\,\dd(Id_2,A)\begin{pmatrix}
		1 & 0 & 0 & 0\\
		0 & 1 & 0 & 0\\
		0 & 0 & 1 & 1\\
		0 & 0 & i &-i\\
	\end{pmatrix}=\begin{pmatrix}
1 & 1 & 1 & 1\\
1 & 1 & -1 & -1\\
1 & -1 & 1 & -1\\
1 & -1 & -1 & 1\\
\end{pmatrix}
\begin{pmatrix} 1&0&0&0\\
0&1&0&0\\
0&0&1&0\\
0&0&\alpha&\beta
\end{pmatrix}\begin{pmatrix}
		1 & 0 & 0 & 0\\
		0 & 1 & 0 & 0\\
		0 & 0 & 1 & 1\\
		0 & 0 & i &-i\\
	\end{pmatrix} .$$
}
\end{rk}

Now we are ready to prove the following result.

\begin{thm}\label{Thm:QAlphaBeta}
Let $M$ be an invertible K80 Markov matrix with eigenvalues $1$, $x$, $y$, $y$ satisfying $x >0 $ and $y \neq x$, {and consider the family of real logarithms $Q(k,\alpha,\beta)$ for $M$ introduced in \eqref{eq:QabDef}. If either $k\neq0$ or $y<0$, an equality $Q(k,\alpha,\beta)=Q(k',\alpha ' , \beta')$} implies $k=k'$, $\alpha=\alpha'$ and $\beta=\beta'$. {Otherwise, if $k=0$ and $y>0$} then $Q(0,\alpha,\beta)=\Log(M)$ for any $\alpha\in \RR, \beta\in \RR_{>0}$.


\end{thm}

\begin{proof}
 {From Remark \ref{rk:VEPs} we have that  $(1,1,1,1)^t$, $(1,1,-1,-1)^t$ are eigenvectors with respective eigenvalues $0$ and $\log(x)$. The other two eigenvalues may be complex:
 \begin{eqnarray}\label{vaps}
 \log|y| \pm \big(2\pi k + \Arg y \big) i,
 \end{eqnarray}
and the corresponding eigenvectors are
 \begin{eqnarray*}
u_+(\alpha,\beta)& := & (1+\alpha,-1-\alpha,1-\alpha, -1+\alpha)^t + i (\beta, -\beta, -\beta,\beta)^t \\ u_-(\alpha,\beta)& := &(1+\alpha,-1-\alpha,1-\alpha, -1+\alpha)^t - i (\beta, -\beta, -\beta,\beta)^t.
 \end{eqnarray*}
}

\noindent {Let us assume that $Q(k,\alpha,\beta) = Q(k',\alpha',\beta')$ for some other choice of $k',\alpha',\beta'$}.

{1st case:  assume that the eigenvalues in (\ref{vaps}) are complex. It is straightforward to check that this is the case if and only if $k\neq 0$ or if $y<0$.   In any case, the four eigenvalues are different and so, the eigenspaces are 1-dimensional. From the equality of eigenvalues for $Q(k,\alpha,\beta)$ and $Q(k',\alpha',\beta')$, we obtain that
\begin{eqnarray*}
\log|y| \pm \big(2\pi k + \Arg y \big) i = \log|y| \pm \big(2\pi k' + \Arg y \big) i.
\end{eqnarray*}
If $2\pi k + \Arg y =2\pi k' + \Arg y $, we derive that $k=k'$. Moreover, the corresponding eigenspaces must be equal, so
$u_+(\alpha,\beta)=\lambda \; u_+(\alpha',\beta')$, for some $\lambda\in \CC\setminus \{0\}$.
From the first and third coordinates of these vectors we have $1+\alpha+\beta i = \lambda(1+\alpha'-\beta' i)$ and $1-\alpha-\beta i = \lambda(1-\alpha'+\beta' i)$ and by summing these two equations we derive that $\lambda = 1$, and $\alpha=\alpha'$ and $\beta=\beta'$.
If  $2\pi k + \Arg y =  -(2\pi k' + \Arg y )$,  then  $u_+(\alpha,\beta)=\lambda \; u_-(\alpha',\beta')$, for some $\lambda\in \CC\setminus \{0\}$. Similarly as above, we deduce that $\lambda = 1$. However, in this case, this implies that  $\beta = -\beta'$, which contradicts the assumption $\beta, \beta' >0$.
}

{2nd case: it remains to deal with the case where the eigenvalues in (\ref{vaps}) are real, that is, $k=0$ and $y>0$.
%
%
In this case, the eigenvalues of $Q(0,\alpha,\beta)$ are the principal logarithms of the eigenvalues of $M$. From the uniqueness of the principal logarithm, we infer that  $Q(0,\alpha,\beta)=Log(M)$ for all $\alpha, \beta$.
%
}
\end{proof}

While the previous results list all the \textit{real} logarithms \textit{with rows summing to zero} of any K80 Markov matrix, we are mainly interested in those that are rate matrices (i.e. we need to restrict to non-negative values in the  off-diagonal entries). The following results characterize the matrices $Q(k,\alpha,\beta)$ that are rate matrices.


\begin{lema}\label{lema:QabIsGen}
With the notation of Theorem \ref{Thm:QAlphaBeta}, we have that $Q(k,\alpha,\beta)$ is a rate matrix if and only if the following inequalities hold:
  \begin{align}
		\log(x) -2 \log|y| & \geq \big| 2\pi k +\Arg y \big| \frac{\big|1 - \alpha^2-\beta^2 \big|}{\beta} \label{eq:IffPhiPsi1}	 \intertext{and}
		-\log(x) &\geq \big| 2\pi k +\Arg y \big| \frac{\big(1+|\alpha|\big)^2 + \beta^2}{\beta}\label{eq:IffPhiPsi2} \, .
  \end{align}
\end{lema}

\begin{proof}
From Proposition \ref{prop:L1L2} we have that $Q(k,\alpha,\beta)$ is a real matrix with rows summing to $0$. Thus we only need to characterize those $Q(k,\alpha,\beta)$ that have non-negative entries outside the diagonal. Let us define $\lambda= -\log(x)$ and $\mu= -\log|y|$. Note that $\lambda,\mu >0 $ because $x,|y|\in (0,1]$ (see Remark \ref{rmk_vaps}). By computing $L_0$ as in Definition \ref{def:L1L2} we get that:
\begin{equation*}
	L_0=1/4
	\begin{small}
		\begin{pmatrix}
			-\lambda -2\mu & -\lambda +2\mu & \lambda & \lambda\\
			-\lambda +2\mu & -\lambda -2\mu& \lambda & \lambda\\
 			\lambda & \lambda & -\lambda -2\mu & -\lambda +2\mu\\
 			\lambda & \lambda & -\lambda +2\mu & -\lambda -2\mu
		\end{pmatrix}
	\end{small} \, .
\end{equation*}
On the other hand, given
	$A=\begin{pmatrix}
		1 & 0\\
		\alpha & \beta\\
	\end{pmatrix} \in GL_2(\RR)$
we obtain the following expression for $L_1(A)$:
\begin{equation*}
	L_1(A) = \frac{1}{4 \beta}
 	\begin{small}
 		\begin{pmatrix}
	  		1 - \alpha^2-\beta^2  & -(1 - \alpha^2-\beta^2) & -\big((1+\alpha)^2 + \beta^2\big) & (1+\alpha)^2 + \beta^2 \vspace*{3mm}\\
	  	 	-(1 - \alpha^2-\beta^2) & 1 - \alpha^2-\beta^2 & (1+\alpha)^2 + \beta^2 & -\big((1+\alpha)^2 + \beta^2\big)\vspace*{3mm}\\
	  	(1-\alpha)^2 + \beta^2 & -\big((1-\alpha)^2 + \beta^2 \big) & -(1 - \alpha^2-\beta^2) & 1 - \alpha^2-\beta^2\vspace*{3mm}\\
	  	-\big((1-\alpha)^2 + \beta^2 \big) & (1-\alpha)^2 + \beta^2 & 1 - \alpha^2-\beta^2 & -(1 - \alpha^2-\beta^2)\\
    \end{pmatrix}
 	\end{small} \, .
\end{equation*}

Recall that $Q(k,\alpha,\beta)= L_0+ \big(2\pi k+\Arg y \big)L_1(A)$. By looking at the off-diagonal entries, we get that $Q(k,\alpha, \beta)$ is a rate matrix if and only if:
$$\begin{aligned}
	-\lambda +2\mu \pm \big(2\pi k +\Arg y \big)\frac{1 - \alpha^2-\beta^2}{\beta}& \geq 0 \qquad (entries \ (1,2) , (2,1) ,(3,4) ,(4,3)\geq 0)\vspace*{3mm}\\
	\lambda \pm \big(2\pi k +\Arg y \big)\frac{(1+\alpha)^2 + \beta^2}{\beta}& \geq 0 \qquad (entries \ (1,3),(1,4),(2,3),(2,4)\geq 0)\vspace*{3mm}\\
	\lambda \pm \big(2\pi k +\Arg y \big)\frac{(1-\alpha)^2 + \beta^2}{\beta}& \geq 0 \qquad (entries \ (3,1),(3,2),(4,1),(4,2)\geq 0) \, .
\end{aligned}$$

The first inequality above gives \eqref{eq:IffPhiPsi1}, while \eqref{eq:IffPhiPsi2} follows by joining the second and third inequalities.
\end{proof}

In the following result, we prove that if $M$ is embeddable then $Q(0,0,1)$ is a Markov generator.

\begin{thm}\label{thm:GLOrt}
Let $M$ be an invertible K80 Markov matrix with eigenvalues $1$, $x$, $y$, $y$ satisfying $x >0 $ and $y \neq x$. Then, if $Q(k,\alpha,\beta)$ is a rate matrix for some $k\in \ZZ$, $\alpha\in \RR$ and $\beta \in \RR_{>0}$ it holds that:
\begin{enumerate}[i)]
	\item $Q(l,\alpha,\beta)$ is a rate matrix for any integer $l \in I_k$ where $$I_k=
	\begin{cases}
 		\langle -k,k \rangle  & \text{if } y>0\\
 		\langle -k-1,k \rangle & \text{if } y<0
	\end{cases} $$
	(we use the notation $\langle a,b\rangle$ to denote the closed interval delimited by $a$ and $b$, no matters if $a>b$ or $a<b$).
	\item $Q( k, 0,1)$ is a rate matrix.
\end{enumerate}
\end{thm}

\begin{proof}\
We will prove that $Q(l,\alpha, \beta)$ and $Q(k,0,1)$ are Markov matrices by checking that they satisfy the inequalities (\ref{eq:IffPhiPsi1}) and (\ref{eq:IffPhiPsi2}) in Lemma \ref{lema:QabIsGen}.
\begin{enumerate}[i)]

	\item  The proof is straightforward from { Lemma \ref{lema:QabIsGen}} and the fact that $| 2\pi l + \Arg y | \leq |2\pi k + \Arg y |$ for any $l \in I_k$ {(note that $\Arg y$ is either 0 if $y>0$ or $\pi$ if $y<0$)}.

	\item Since $Q(k,\alpha,\beta)$ is a rate matrix it follows from Lemma \ref{lema:QabIsGen} that:
 	\begin{align*}
		\log(x) - 2\log(y) & \geq \big| 2\pi k +\Arg y \big| \frac{ \big| 1 - \alpha^2-\beta^2 \big| }{\beta}\geq {0 = \big| 2\pi k +\Arg y \big| \frac{\big|1 - 0^2-1^2 \big|}{1} } \intertext{and}
		-\log(x) & \geq \big| 2\pi k +\Arg y \big| \frac{\big(1+|\alpha|\big)^2 + \beta^2}{\beta} \geq \big| 2\pi k +\Arg y \big| \frac{\big(1 + 0)^2 + \beta^2}{\beta} \, .
 	\end{align*}
	Now, let us consider the real function $f(\beta)=\frac{1+\beta^2}{\beta}$ restricted to $\RR_{>0}$. $f$ is continuous and a straightforward computation shows that $f$ has an absolute minimum at $\beta =1$. This concludes the proof.
\end{enumerate}
\end{proof}


Now we are ready to prove the main result in this section, which characterizes embeddable K80 Markov matrices.
\begin{cor}\label{cor:CharOfEmb}
A K80 Markov matrix $M$ with eigenvalues $1$, $x$, $y$, $y$ is embeddable if and only if $Q(0,0,1)$ is a rate matrix, where
$$Q(0,0,1) =S\;
	\begin{small} \begin{pmatrix}
		0 & 0 & 0 & 0\\
		0 & \log(x) & 0 & 0\\
		0 & 0 & \log|y|& -\Arg y \\
	0 & 0 & \Arg y & \log|y|\\
	\end{pmatrix} \end{small}
\;S^{-1}. $$
In particular, if $y>0$ then $M$ is embeddable if and only if $\Log(M)$ is a rate matrix.

\end{cor}
\begin{proof}
%
If $x<0$ or $\det(M)=0$,  then $M$ is not embeddable (see Remark \ref{rmk_vaps}).
For $x=y>0$, $M$ is a JC69 matrix and it is known that a JC69 matrix is embeddable if and only if its eigenvalues are positive \citep{JJ}. In this case, $y>0$ and $ Q(0,0,1) = \Log(M)$ is a rate matrix (see Theorem \ref{thm:LogK3Embed}).
Finally, if $x \neq y $ and $x>0$ the first claim follows from  Proposition \ref{prop:L1L2}, Remark \ref{rk:AtoAlphaBeta} and Theorem \ref{thm:GLOrt}.
To conclude the proof, note that if $y>0$ then $\Arg y=0$ thus $Q(0,0,1) = \Log(M)$ (see Definition \ref{def:log}).
\end{proof}

The following corollaries use Corollary \ref{cor:CharOfEmb} to characterize embeddable K80 Markov matrices in terms of its eigenvalues (Corollary \ref{cor:VapsEmbed}) and in terms of its entries (Corollary \ref{cor:EntriesEmbed}).

\begin{cor}\label{cor:VapsEmbed}
Let $M$ be a K80 Markov matrix with eigenvalues $1$, $x$, $y$, $y$. Then:
\begin{enumerate}[(i)]
	\item If $y=0:$ $M$ is not embeddable.
	\item If $y>0:$ $M$ is embeddable if and only if $x\geq y^2$. In this case, $Log(M)$ is a rate matrix.
	\item If $y<0:$ $M$ is embeddable if and only if $e^{-2\pi} \geq x \geq y^2$.
\end{enumerate}
\end{cor}

\begin{proof}
By Corollary \ref{cor:CharOfEmb} we know that $M$ is embeddable if and only if $Q(0,0,1)$ is a rate matrix. If $x\neq y$ it follows from Lemma \ref{lema:QabIsGen} that $Q(0,0,1)$ is a rate matrix if and only if $\log(x) - 2\log|y| \geq 0$ and $-\log(x) \geq 2 \Arg y $. {By computing the exponential of both sides of the expression $\log(x) \geq 2\log|y|$ we get $x\geq y^2$. Now, if $y<0$ we have $\Arg{y}=\pi$ and hence $-log(x) \geq 2 \Arg y$ implies $e^{-2\pi}\geq x$. On the other hand, if $y>0$ we have that $\Arg{y}=0$ and hence $-\log(x) \geq 2 \Arg y $ implies that $1 \geq x$ . Note that this last constraint is redundant since the absolute value of the eigenvalues of any Markov matrix is bounded by 1 due to Perron-Frobenius theorem.}
\end{proof}

\begin{rk}\label{rmk_Log}
\rm In the case of positive eigenvalues $x$ and $y$, we have that $M$ is embeddable if and only if its principal logarithm $\Log(M)$ is a rate matrix (Corollary \ref{cor:VapsEmbed}). However, when the multiple eigenvalue $y$ is negative, $\Log(M)$ is never a rate matrix. In this case, the embeddability cannot be checked by looking at $\Log(M)$ and the previous corollary provides an embeddability criterion.
{In particular, the third item in the corollary above (or in Corollary \ref{cor:EntriesEmbed} below) provides a non-zero measure subset formed by embeddable K80 matrices whose principal logarithm is not a rate matrix.}
\end{rk}

\begin{cor}\label{cor:EntriesEmbed}
Let $M=K(1-b-2c,b,c,c)$. Then:
\begin{enumerate}[(i)]
	\item If ${2}c=1-2b$, $M$ is not embeddable.
	\item If ${2}c<1-2b$, $M$ is embeddable if and only if $c \leq \sqrt{b} -b \ (\leq \frac{1}{4})$. In this case, $Log(M)$ is a rate matrix.
	\item If ${2}c>1-2b$, $M$ is embeddable if and only if $\frac{1}{4} - \frac{e^{-2\pi}}{4} \leq c \leq \sqrt{b} -b \ (\leq \frac{1}{4})$.
\end{enumerate}
\end{cor}

\begin{proof}
The claim follows directly from Corollary \ref{cor:VapsEmbed} by expressing the eigenvalues in terms of the entries using the bijection given in (\ref{eq:paramBijection}).
\end{proof}

\begin{rk}\label{rmk_embedk81}
\rm
As claimed in the introduction, the results above together with Corollary 3.5 by \citet{JJ} solve the embedding problem for any K81 matrix. Indeed, the embeddability of any K81 matrix with two repeated eigenvalues can be treated analogously to K80 matrices (one just needs to permute rows and columns accordingly). If the matrix has three repeated eigenvalues, then it is a JC69 matrix and the following argument solves the embeddability of JC69 matrices: there are no embeddable matrices with negative determinant \citep{Culver} and by Theorem \ref{thm:LogK3Embed}, the principal logarithm of a JC69 matrix with positive determinant is always a rate matrix. Thus, a JC69 matrix is embeddable if and only if its determinant is positive.
\end{rk}

\section{Identifiability of rates for K80 Markov matrices}\label{sec:identifiability}

In this section we address the identifiability problem for K80 embeddable matrices. As a consequence of the results obtained in the previous section we provide a criterion to determine whether the rates of these matrices are identifiable or not. Furthermore, for those matrices with non-identifiable rates we determine how many Markov generators they admit.

\begin{prop} \label{prop:identIFF}
Let $M$ be an embeddable K80 Markov matrix with eigenvalues $1$, $x$, $y$, $y$ {and $x\neq y$}. Then, the rates of $M$ are identifiable if and only if $Q(-1,0,1)$ is not a rate matrix.
\end{prop}

\begin{proof}
By Corollary \ref{cor:CharOfEmb} we know that $M$ is embeddable if and only if $Q(0,0,1)$ is a rate matrix.

Now assume that there are $\alpha\in \RR$, $\beta \in \RR_{>0}$ and $k\in\ZZ$ such that $Q(k,\alpha,\beta)$ is a Markov generator for $M$ different {from} $Q(0,0,1)$. In this case, Theorem \ref{thm:GLOrt} gives that $Q(-1,0,1)$ is also Markov generator because $-1$ belongs to the interval $I_{k}$, independently of the sign of $y$ (note that if $y>0$, the case $k=0$ is excluded by Theorem \ref{Thm:QAlphaBeta}). Moreover, according to Theorem \ref{Thm:QAlphaBeta}, $Q(0,0,1)$ and $Q(-1,0,1)$ are distinct Markov generators.
\end{proof}

\begin{ex}\label{ex:pos_eig}
In this example we show an embeddable K80 Markov matrix with positive eigenvalues and non-identifiable rates. Let us consider $M$ the K80 Markov matrix with eigenvalues $1$, $x=e^{-4\pi}$ and $y=e^{-2\pi}$ (with multiplicity 2). Rounding to the 10th decimal the entries of $M$ are:
$$M= \begin{small} \begin{pmatrix}
		0.2509345932 & 0.2490671504 & 0.2499991282 & 0.2499991282\\
		0.2490671504 & 0.2509345932 & 0.2499991282 & 0.2499991282\\
		0.2499991282 & 0.2499991282 & 0.2509345932 & 0.2490671504\\
		0.2499991282 & 0.2499991282 & 0.2490671504 & 0.2509345932\\
  \end{pmatrix} \end{small} \, .$$
A straightforward computation shows that $\Log(M)$ is a Markov generator and hence $M$ is embeddable:
$$\Log(M)= \begin{small} \begin{pmatrix}
 		-2\pi & 0 & \pi & \pi\\
 		0 & -2\pi & \pi & \pi\\
 		\pi & \pi & -2\pi & 0 \\
 		\pi & \pi & 0 & -2\pi \\
  \end{pmatrix} \end{small}\, .$$
Nonetheless, the rates of $M$ are not identifiable since there are other Markov generators for it:
$$Q(-1,0,1)= \begin{small} \begin{pmatrix}
 		-2\pi & 0 & 2\pi & 0\\
 		0 & -2\pi & 0 & 2\pi\\
 		0 & 2\pi & -2\pi & 0 \\
		2\pi & 0 & 0 & -2\pi \\
  \end{pmatrix} \end{small}
\qquad \text{ and } \qquad Q(1,0,1) =
	\begin{small} \begin{pmatrix}
 		-2\pi & 0 & 0 & 2\pi\\
 		0 & -2\pi & 2\pi & 0\\
 		2\pi & 0 & -2\pi & 0 \\
 		0 & 2\pi & 0 & -2\pi \\
  \end{pmatrix} \end{small}\, .$$

{Note that this matrix has determinant equal to $e^{-8\pi}$ and is not close to the identity matrix $Id_4$. Indeed, in the \emph{Frobenius distance}, $\|M-Id_4\|_F \approx 1.729893126$. At the other extreme, for any positive K80 matrix $K$, the limiting matrix $H:=\lim_{n\rightarrow \infty}K^n$  is the $4\times4$ matrix whose entries are all $1/4$ and for the matrix above we have $\|M-H\|_F \approx 0.002640965$ (so $M$ is closer to $H$ than to $Id_4$).}
\end{ex}

\begin{ex}\label{ex:negative}
Here we show an embeddable K80 Markov matrix with some negative eigenvalues and non-identifiable rates. Let us consider $M$ the K80 Markov matrix with eigenvalues $1$, $x=e^{-2\pi}$ and $y=-e^{-\pi}$ (with multiplicity 2). Rounding to the 10th decimal  the entries of $M$ are:
$$M=\begin{small}
	\begin{pmatrix}
		0.2288599016 & 0.2720738198 & 0.2495331393 & 0.2495331393\\		
		0.2720738198 & 0.2288599016 & 0.2495331393 & 0.2495331393\\
		0.2495331393 & 0.2495331393 & 0.2288599016 & 0.2720738198\\
		0.2495331393 & 0.2495331393 & 0.2720738198 & 0.2288599016\\
	\end{pmatrix}
\end{small} \, .$$
As we can see, $\Log(M)$ is not a real matrix:
$$ \Log(M)=\frac{1}{2}
\begin{small}
	\begin{pmatrix}
  		-2 \pi +\pi \ i & -\pi\ i 	 & \pi & \pi\\
 		-\pi\ i   & -2 \pi +\pi \ i & \pi & \pi\\
 		\pi & \pi & -2 \pi +\pi \ i & -\pi\ i\\
 		\pi & \pi & -\pi\ i & -2 \pi +\pi \ i\\
 \end{pmatrix}
\end{small}\, .$$
In spite of that, $Q(0,0,1)$ is a rate matrix, so $M$ is embeddable. Furthermore, $Q(-1,0,1)$ is also a Markov generator for $M$ and hence the rates are not identifiable:
\begin{center}
	$ Q(0,0,1)=
	\begin{footnotesize}
		\begin{pmatrix}
  			- \pi & 0 	 & 0 & \pi\\
 			0  & - \pi  & \pi & 0\\
 			\pi & 0 & - \pi & 0\\
 			0 & \pi & 0 & - \pi\\
 		\end{pmatrix}
	\end{footnotesize}$
	\vspace*{3mm} \qquad
	 $Q(-1,0,1)=
	 \begin{footnotesize}
		\begin{pmatrix}
 			- \pi & 0	 & \pi & 0\\
 			0  & - \pi  & 0 & \pi\\
 			{0} & {\pi} & - \pi & 0\\
 			{\pi} & {0} & 0 & - \pi\\
 		\end{pmatrix}
	\end{footnotesize}.$
\end{center}

Unlike the previous example, in this case $Q(1,0,1)$ is not a rate matrix. Furthermore, $M$ is a K80 embeddable matrix with no K80 Markov generators (because $\Log(M)$ is not a rate matrix, see Theorem \ref{thm:LogK3Embed}), {which shows that rates do not necessarily satisfy the same symmetry constraints as probabilities do \citep{kimura80}; see also \citep{JJ}}.

{With the notation of the previous example, note that this matrix has $\|M-Id_4\|_F \approx 1.781252133$, $\det(M) =e^{-4 \pi}$, and $\|M-H\|_F \approx 0.06114223420$.}
\end{ex}


\begin{figure}[h]
  \centering
  \begin{tabular}{cc}
    \includegraphics[width=7cm]{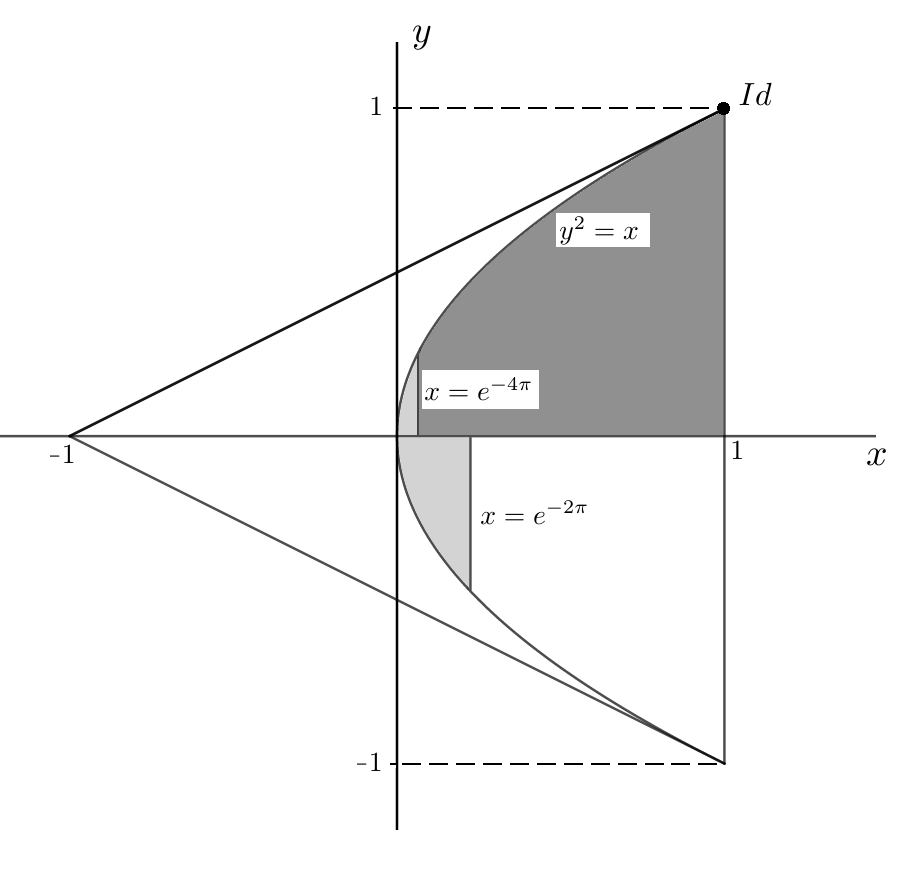} & \includegraphics[width=7cm]{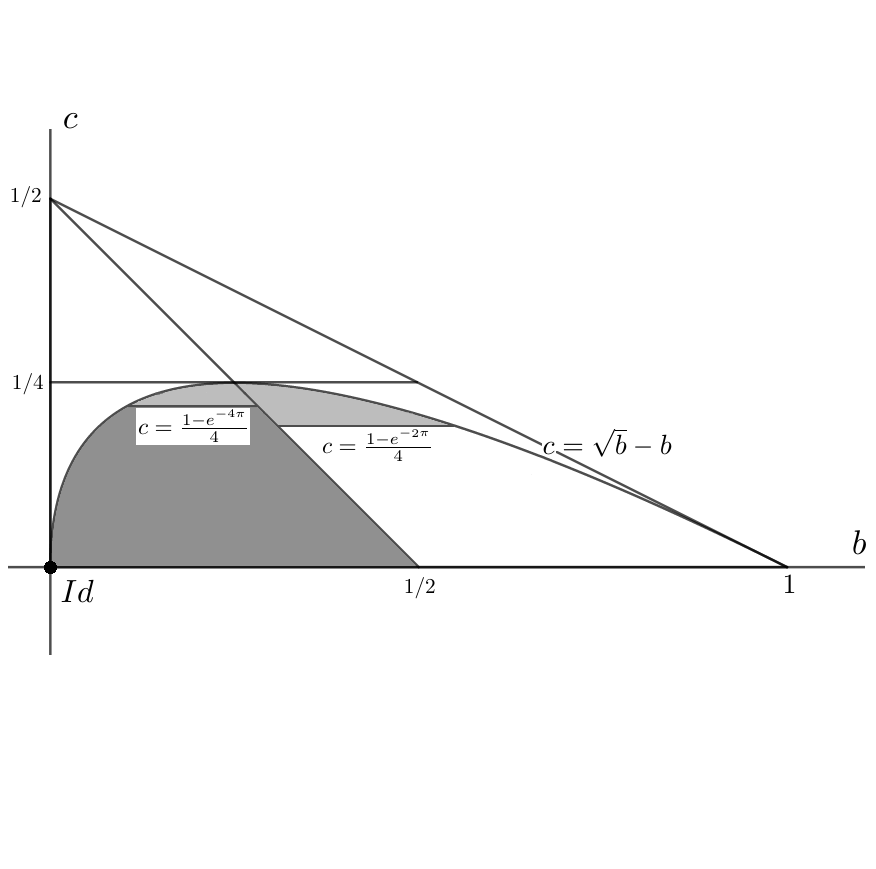} \\
    (a) Parametrization in terms of eigenvalues $x$ and $y$ & (b) Parametrization in terms of entries $b$ and $c$
    \end{tabular}
  \caption{Parameterizations of K80 Markov matrices in terms of eigenvalues (a) and entries (b): embeddable matrices with only one Markov generator in dark grey and embeddable matrices with infinitely many Markov generators in light grey. The lines separating both areas contain embeddable matrices with non-identifiable rates but a finite number of Markov generators. The scale in these figures is not exact so that the light grey subset could be visualized.}\label{fig:Parametrizations}
\end{figure}

\begin{rk}
\rm In the examples \ref{ex:pos_eig} and \ref{ex:negative}, the Markov generators other than the principal logarithm are not K81 matrices, they belong to one of the Lie Markov models listed by \citet{LieMM}, namely the model 3.3b. This is another 3-dimensional model, different from the K81 model, which contains the K80 model as well.
\end{rk}

\begin{thm}\label{thm:vapsIdentiff}

Let $M$ be an embeddable K80 Markov matrix with eigenvalues $1$, $x$, $y$, $y$ {and $x\neq y$}. Then the following holds:
\begin{enumerate}[a)]

	\item If $y>0$ and $x > e^{-4\pi}$ then $M$ has only one Markov generator, which is its principal logarithm.

	\item If $y>0$ and $x=e^{-4\pi}$ then $M$ has exactly $3$ generators: $Q(0,0,1)$ (which coincides with $\Log(M)$), $Q(1,0,1)$ and $Q(-1,0,1)$.

	\item If $y<0$ and $x=e^{-2\pi}$ then $M$ has exactly $2$ generators: $Q(0,0,1)$ and $Q(-1,0,1)$.
\end{enumerate}
{Otherwise, $M$ has infinitely many Markov generators.}
\end{thm}

Figure \ref{fig:Parametrizations} illustrates Theorem \ref{thm:vapsIdentiff} and Theorem \ref{thm:entriesIdentiff}.

\begin{proof}
For ease of reading let us consider $\phi:\RR\times \RR_{>0} \rightarrow \RR$, defined by $\phi(\alpha, \beta)= \frac{\big(1+|\alpha|\big)^2+\beta^2}{\beta}$. In the proof of Theorem \ref{thm:GLOrt} we already saw that $\phi$ has an absolute minimum at $\phi(0,1)=2$.

{Note that, since $M$ is embeddable, we have $x \geq y^2$ (Corollary \ref{cor:VapsEmbed}) and hence inequality (\ref{eq:IffPhiPsi1}) is satisfied for any $\alpha,\beta$ such that $\alpha^2+\beta^2 =1$ or for $k=0$ when $y>0$.}

{Let $Q(k,\alpha, \beta)$ be a Markov generator for $M$.}
\begin{enumerate}[a)]
	\item We know that $Q(-1,0,1)$ is not a rate matrix by Lemma \ref{lema:QabIsGen} {(it does not satisfy (\ref{eq:IffPhiPsi2}))}. Hence, by Proposition \ref{prop:identIFF} the rates of $M$ are identifiable. Furthermore, due to Corollary \ref{cor:CharOfEmb} we have that the only Markov generator of $M$ must be its principal logarithm.

	\item {Since $y>0$ we have that $\Arg{y}=0$. Hence,} it follows from inequality (\ref{eq:IffPhiPsi2}) in Lemma \ref{lema:QabIsGen} that $4\pi \geq 2|k|\pi\; \phi(\alpha,\beta)$. Using that $\phi(\alpha,\beta)>2$ for $(\alpha,\beta)\neq(0,1)$ we get that the inequality holds if and only if $k=0$ (and hence $Q(0,\alpha,\beta)=\Log(M)$) or $\alpha=0$, $\beta=1$ and $|k| = 1$. { Note that all these solutions do also satisfy inequality (\ref{eq:IffPhiPsi1}).}

	\item {Since $y<0$ we have that $\Arg{y}=\pi$. Hence,} it follows from inequality (\ref{eq:IffPhiPsi2}) in Lemma \ref{lema:QabIsGen} that $2\pi \geq \big|(2k+1)\pi\; \phi(\alpha,\beta) \big|$. Using that $\phi(\alpha,\beta)>2$ for $(\alpha,\beta)\neq(0,1)$ we get that the inequality holds if and only if $k=0$ \big(and hence $Q(0,\alpha,\beta)=\Log(M)$\big) or $\alpha=0$, $\beta=1$ and $|2k+1| \leq 1$. { Note that all these solutions do also satisfy inequality (\ref{eq:IffPhiPsi1}).}
\end{enumerate}
{If $M$ does not lie in the already covered cases, note that if $y>0$ we have $x< e^{-4\pi}$. In particular,} $-\log(x) > 4\pi= 2\pi|-1|\; \phi(0,1)$ and inequality (\ref{eq:IffPhiPsi2}) is satisfied for $Q(-1,0,1)$. Furthermore, inequality (\ref{eq:IffPhiPsi2}) is satisfied for any $(\alpha,\beta)$ close enough to $(0,1)$. Hence, it follows from Lemma \ref{lema:QabIsGen} that $Q(-1,\alpha,\beta)$ is a Markov generator of $M$ for any $(\alpha,\beta)$ close enough to $(0,1)$ such that $\alpha^2+\beta^2=1$. The same argument does also work for $y<0$.
\end{proof}


\begin{rk} \label{rk:BranchLength}
\rm  The determinant of $M$ is related to the expected number of nucleotide substitutions in the Markov process ruled by $M$. In a phylogenetic tree, the length of a branch representing an evolutionary process between an ancestral species and a descendant species is usually measured as the expected number of nucleotide substitutions per site. If this process is ruled by a {substitution} matrix with uniform stationary distribution 
 {then this expected number of substitutions can be approximated by $l(M):=-\frac{1}{4}log (\det(M))$ \citep[for a precise formulation see][]{barry_hartigan}. For example, for the identity matrix one has $l(Id_4)=0$ and for the limiting matrix $H$ introduced in Example \ref{ex:pos_eig}, $l(H)=\infty$. 
 We come back to this biological concept in the Discussion but for the moment we give conditions on the identifiability in terms of $l(M)$.}
\end{rk}

{
 \begin{cor}\label{cor:brlength}
Let $M$ be an embeddable K80 Markov matrix with eigenvalues $1$, $x$, $y$, $y$ {with $y\neq x>0$}. Then,
\begin{enumerate}
 \item[(a)] If $y>0$ and $l(M) < {2\pi}$, the rates of $M$ are identifiable.
 \item[(b)] If $y<0$ and $l(M) < {\pi}$, the rates of $M$ are identifiable.
\end{enumerate}
Moreover, these bounds are sharp (tight examples are given in \ref{ex:pos_eig} with $l(M)=2\pi$ and in \ref{ex:negative} with $l(M)=\pi$).
 \end{cor}
}

\begin{proof}
{Theorem \ref{thm:vapsIdentiff} shows that the eigenvalues $x,y$ of embeddable K80 matrices must satisfy $x \geq y^2$ and hence their determinant $xy^2$ must be smaller than or equal to  $x^2$. On the other hand, Theorem \ref{thm:vapsIdentiff} provides a bound on $x$ (depending on the sign of $y$) to determine whether an embeddable matrix has identifiable rates or not. More precisely, K80 embeddable matrices with non-identifiable rates have determinant at most $e^{-8\pi}$ if all the eigenvalues are positive or determinant at most $e^{-4\pi}$ if they have a repeated negative eigenvalue. Equivalently, their branch length is greater than or equal to $2\pi$ or $\pi$, respectively.}

{As we noted in the Examples \ref{ex:pos_eig} and \ref{ex:negative} above, the K80 embeddable matrices provided there have non-identifiable rates and determinant $e^{-8\pi}$ and $e^{-4\pi}$, respectively.} 
\end{proof}

{In the following corollary we restate Theorem \ref{thm:vapsIdentiff} in terms of the entries of the Markov matrix.}

\begin{cor}\label{cor:entriesIdentiff2}
{Let $M=K(1-b-2c,b,c,c)$ with $b\neq c$ be an embeddable K80 Markov matrix. Then the following holds:}
\begin{enumerate}[i)]

	\item { If $2c<1-2b$ and $c <\frac{1}{4}-\frac{ e^{-4\pi}}{4}$ then $M$ has only one Markov generator, which is its principal logarithm.}

	\item { If $2c<1-2b$ and  $c = \frac{1}{4}-\frac{ e^{-4\pi}}{4}$ then $M$ has exactly $3$ generators: $Q(0,0,1)$ (which coincides with $\Log(M)$), $Q(1,0,1)$ and $Q(-1,0,1)$.}

	\item { If $2c>1-2b$ and $c = \frac{1}{4}-\frac{ e^{-2\pi}}{4}$ then $M$ has exactly $2$ generators: $Q(0,0,1)$ and $Q(-1,0,1)$.}
\end{enumerate}
{Otherwise, $M$ has infinitely many Markov generators.}
\end{cor}

\begin{proof}
{ The claim follows directly from Theorem \ref{thm:vapsIdentiff}  by expressing the eigenvalues in terms of the entries using the bijection given in (\ref{eq:paramBijection}).}
\end{proof}

{We can now provide a proof for Theorem \ref{thm:entriesIdentiff}:}

\begin{proofThm11}
{Theorem \ref{thm:entriesIdentiff} is the summary of the results in Corollary \ref{cor:EntriesEmbed} and Corollary \ref{cor:entriesIdentiff2}.}
\end{proofThm11}

 {We turn our attention to matrices whose diagonal entries are the largest entries in each column. These are called \emph{diagonal largest in column} matrices, briefly \emph{DLC}, and are related to matrix parameter identifiability in phylogenetics \citep[see][]{Chang}. A straightforward computation shows that
 \begin{equation}\label{eq:DLC}
M=K(1-b-2c,b,c,c) \textrm{ is DLC if and only if } 2b+2c<1, \, b+3c<1.
 \end{equation}
}
Note that the Markov matrix of Example \ref{ex:pos_eig} is DLC, thus there are DLC matrices whose rates are not identifiable.

\begin{cor}\label{DLC}{The following holds:}
	\begin{enumerate}[(a)]
		\item {A K80 DLC matrix is embeddable if and only if its principal logarithm is a rate matrix.}
		\item { There is an open set of embeddable DLC matrices whose rates are not identifiable.}
	\end{enumerate}
\end{cor}
\proof{ { It follows from the first inequality in \eqref{eq:DLC} that DLC K80 Markov matrices lie in case (b) of Theorem \ref{thm:entriesIdentiff}.  The second statement follows from case (b) iii) in the same Theorem. }}

\section{{Measuring the set of embeddable and rate identifiable matrices}} \label{sec:Volumes}

In this section we will study how many K80 Markov matrices are embeddable and also how many of those matrices have identifiable rates. We will solve these questions when restricted to some well defined subsets of K80 matrices. We will proceed by using a parametrization of all K80 matrices and computing the volumes {(actually areas)} needed in that set. We will consider the following regions:
\begin{enumerate}[i)]
	\item $\Delta$: The set of all $K80$ Markov matrices.

	\item $\Delta_{\rm{DLC}}$: The set of all DLC matrices in $\Delta$ (see Remark \ref{DLC}).

	\item $\Delta_+$: The set of all $M \in \Delta$ in the connected component of the identity with respect to vanishing determinant (that is, all K80 Markov matrices with only positive eigenvalues).
	
	\item $\Delta_{\rm{dd}}$: The set of all $M \in \Delta$ such that the probability of not mutating is higher than the probability of mutating, i.e. those matrices such that $a> b+2c$. These matrices are said to be \textit{diagonally-dominant} and it is well known that they have identifiable rates (if embeddable) \citep[see][]{Cuthbert72}.
\end{enumerate}


In the previous sections we have been using two different parameterizations of the model (see { Equation (\ref{eq:paramBijection}) and } Figure \ref{fig:Parametrizations}), one in terms of the eigenvalues of the Markov matrix ($x$ and $y$) and the other in terms of its entries ($b$ and $c$) . { The first parametrization can be used to easily describe $\Delta_+$ whereas the second provides an easier description of $\Delta_{\rm{DLC}}$ and $\Delta_{\rm{dd}}$. Using these two parametrizations, we obtain that $\Delta \supset \Delta_{\rm{DLC}} \supset \Delta_+ \supset \Delta_{\rm{dd}} $. Figures \ref{fig:SetsVaps} and  \ref{fig:SetsEntries} illustrate these inclusions in terms of the eigenvalues and the entries of the Markov matrix, respectively. }


\begin{figure}[h]
  \centering
  \begin{tabular}{cccc}
    \includegraphics[width=4cm]{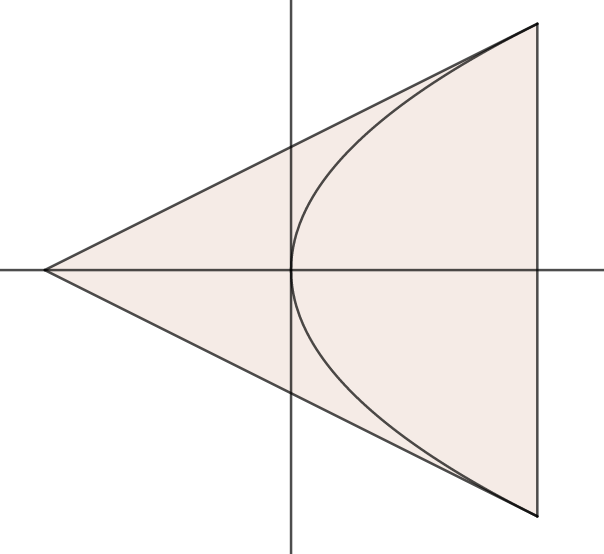}& 
    \includegraphics[width=4cm]{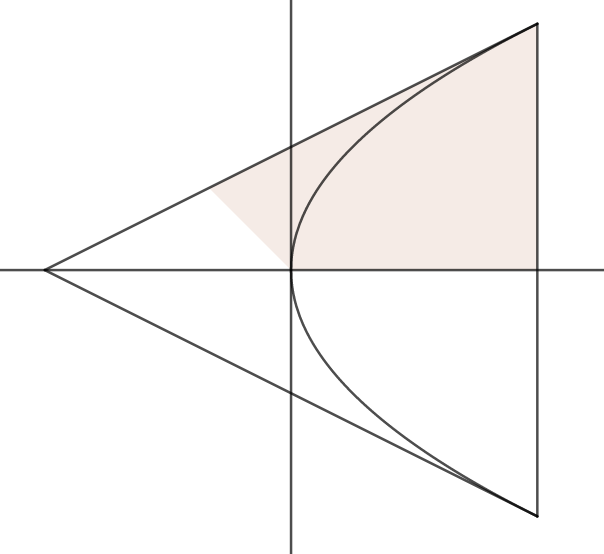}&
      \includegraphics[width=4cm]{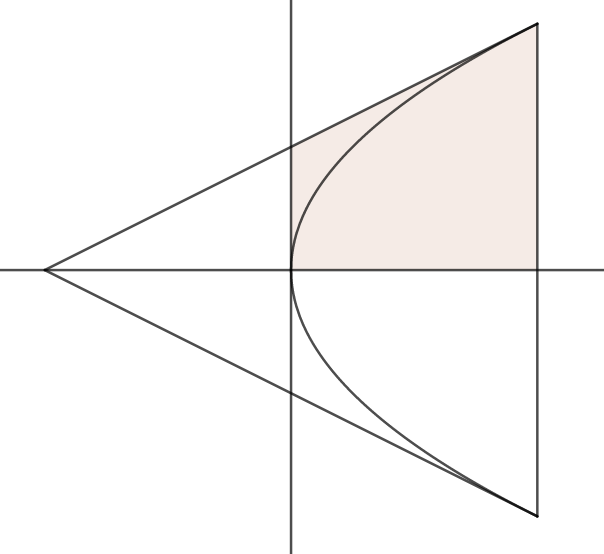}&
    \includegraphics[width=4cm]{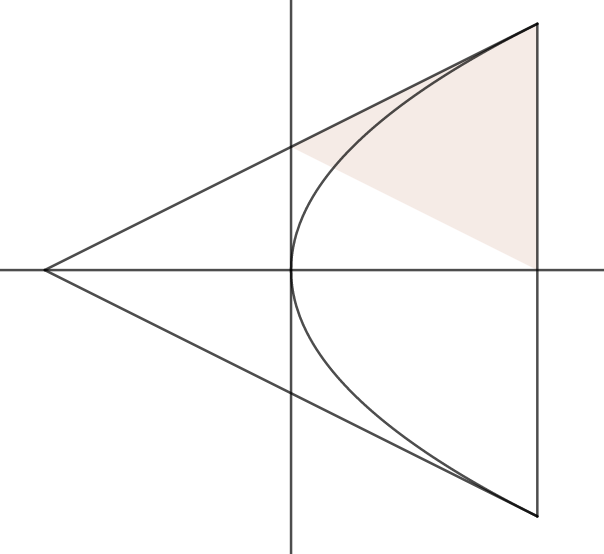}\\
    (a) $\Delta$& (b) $\Delta_{\rm{DLC}}$& (c) $\Delta_+$& (d) $\Delta_{\rm{dd}}$
    \end{tabular}
  \caption{Subsets parametrized in terms of the eigenvalues.}\label{fig:SetsVaps}
\end{figure}


\begin{figure}[h]
  \centering
  \begin{tabular}{cccc}
   \includegraphics[width=4cm]{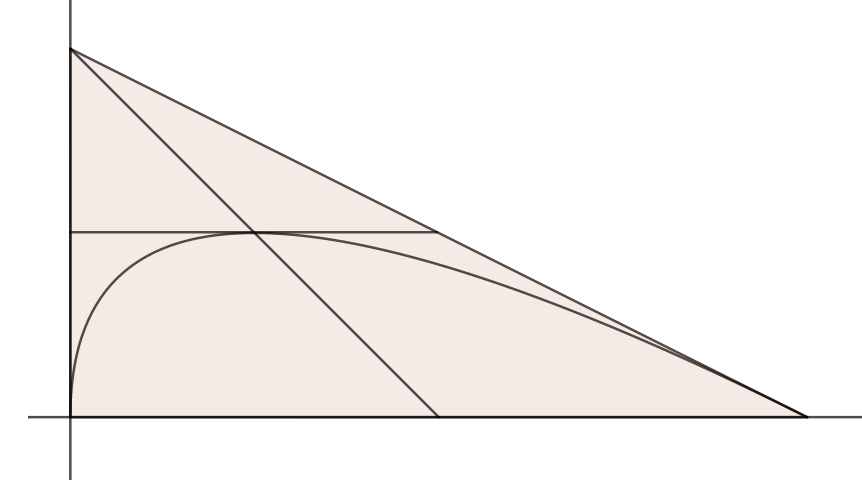}&
      \includegraphics[width=4cm]{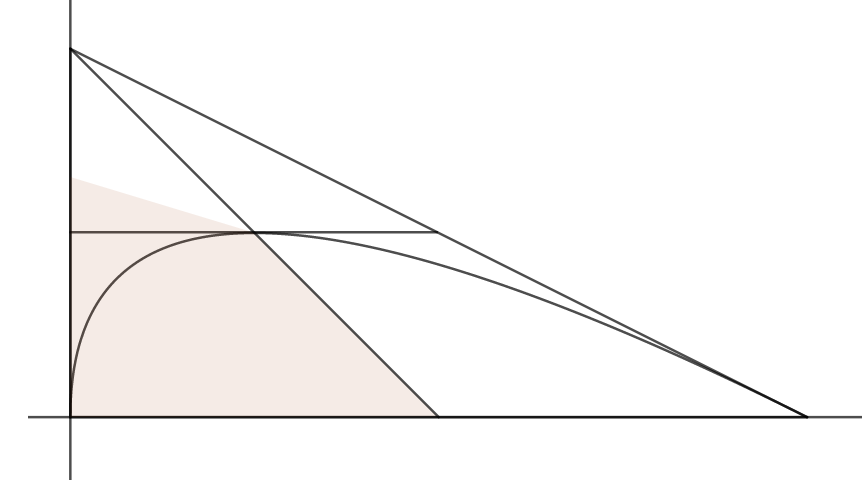}&
  \includegraphics[width=4cm]{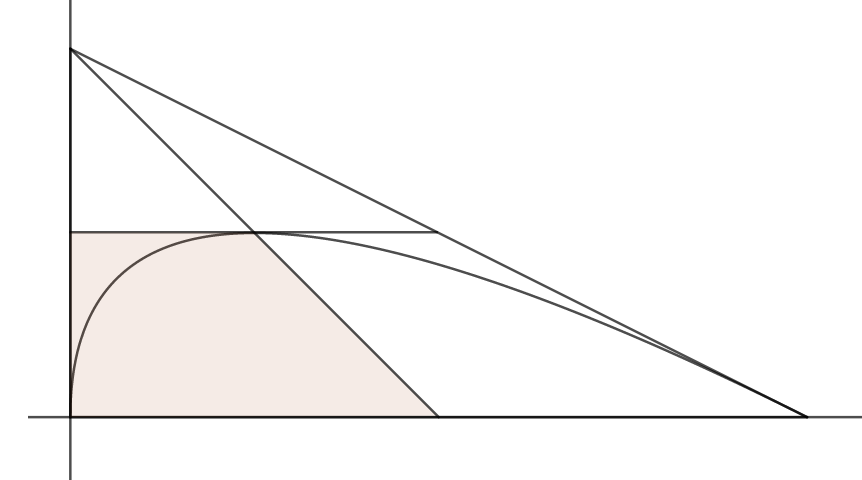}&
  \includegraphics[width=4cm]{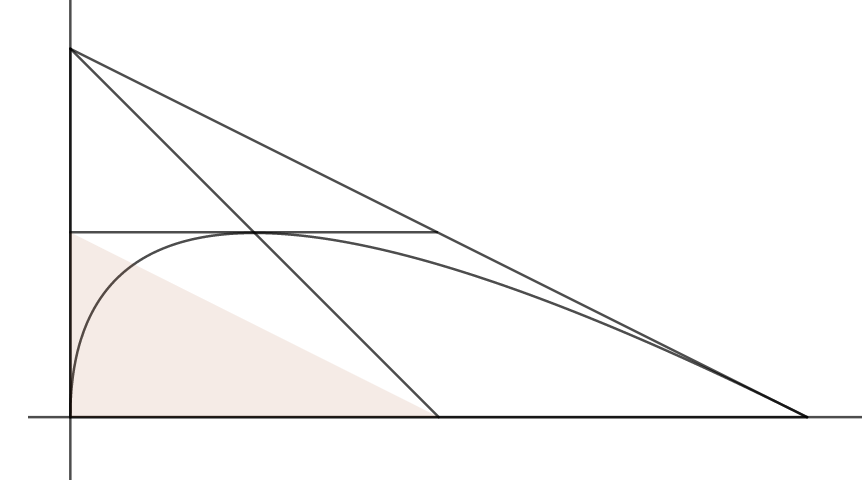}\\
      (a) $\Delta$& (b) $\Delta_{\rm{DLC}}$& (c) $\Delta_+$& (d) $\Delta_{\rm{dd}}$
   \end{tabular}
 \caption{Subsets parametrized in terms of the entries.}\label{fig:SetsEntries}
\end{figure}

Note that, since the {bijection $\varphi$ in \eqref{eq:paramBijection}} is a linear map, the relative volumes of embeddable matrices will not depend on the chosen parameters (entries or eigenvalues). Actually, $\det(D(\varphi))= 8$, so the volume of any subset will be eight times greater in the eigenvalues' parametrization. We decided to use the parametrization in terms of the eigenvalues to compute the volumes we want to know because the expressions appearing {in} the computations are simpler. \\

For a {clearer} picture of embeddability and rate identifiability of those subsets of K80 matrices, intersect Figure \ref{fig:Parametrizations} and Figures \ref{fig:SetsVaps} and \ref{fig:SetsEntries}.

\begin{prop}
Let $\Delta_{\rm{emb}}$ denote the set of embeddable K80 Markov matrices, and let $\Delta_{\rm{idf}}\subset \Delta_{\rm{emb}}$ be the subset of matrices with identifiable rates. Using the parametrization of K80 Markov matrices in terms of its eigenvalues, the following holds: a) $V(\Delta)=2$; b) $V(\Delta_{\rm{DLC}})= 10/12$; c) $V(\Delta_+)= 3/4$; d) $V(\Delta_{\rm{dd}}) = 1/2$; e) $V(\Delta_{\rm{emb}})= \frac{2(1+e^{-3\pi})}{3}$; f) $V(\Delta_{\rm{idf}})= \frac{2(1-e^{-6\pi})}{3}$; g) $V(\Delta_{\rm{emb}} \cap \Delta_{\rm{DLC}})= V(\Delta_{\rm{emb}} \cap \Delta_+)= 2/3$; h) $V(\Delta_{\rm{emb}} \cap \Delta_{\rm{dd}})= V(\Delta_{\rm{idf}} \cap \Delta_{\rm{dd}})=\frac{7-4\sqrt{2}}{3} $.
\end{prop}

\begin{proof}
\begin{enumerate}[a)]
	\item $\Delta$ is the triangle with vertices $(-1,0), \ (1,1) \text{ and } (1,-1)$ which has area 2.
	
	\item $\Delta_{\rm{DLC}}$ is the polygon with vertices $(-1/3,1/3), \ (1,1), \ (1,0) \text{ and } (0,0)$ which has area 10/12.

	\item $\Delta_+$ is the trapezoid with vertices $(0,1/2), \ (1,1), \ (1,0) \text{ and } (0,0)$ which has area 3/4.

	\item $\Delta_{\rm{dd}}$ is the triangle with vertices $(0,1/2), \ (1,1) \text{ and } (1,0)$ which has area 1/2.

	\item It follows from Corollary \ref{cor:VapsEmbed} that $V(\Delta_{\rm{emb}}) = \displaystyle \int_0^1 \int^1_{y^2} 1 \ dx dy \ + \ \ \int_{-e^{-\pi}}^0 \int_{y^2}^{e^{-2\pi}} 1 \ dx dy = \frac{2(1+e^{-3\pi})}{3}$.

	\item It follows from Theorem \ref{thm:vapsIdentiff} that $V(\Delta_{\rm{idf}})= \displaystyle \int_0^1 \int^1_{y^2} 1 \ dx dy \ - \ \ { \int_0^{e^{-2\pi}} } \int_{y^2}^{e^{-4\pi}} 1 \ dx dy =\frac{2(1-e^{-6\pi})}{3}$.

	\item {As shown in the proof of Corollary \ref{DLC}, the entries of a DLC K80 Markov matrix must satisfy $1-b-2c>b$ and $1-b-2c>c$. Hence, the repeated eigenvalue of a K80 DLC Markov matrix $y=1-2b-2c$ (see equation \eqref{eq:eigenvalues}), must be positive. Since embeddable matrices have positive determinant, we have} that $\Delta_{\rm{emb}} \cap \Delta_{\rm{DLC}}= \Delta_{\rm{emb}} \cap \Delta_+$ (see Figure \ref{fig:SetsVaps}). It follows from Corollary \ref{cor:VapsEmbed} that the volume of this set can be computed as $\displaystyle \int_{0}^1 \sqrt{x} \ dx =2/3$.

	\item The set of diagonally-dominant matrices is the triangle with vertices $(0,0.5), \ (1,1)$ and $(1,0)$. It is known that diagonally-dominant Markov matrices have only one real logarithm \citep{Cuthbert72} and hence the first equality follows. Furthermore the points where the curve $y^2=x$ intersects with the boundary of $\Delta_{\rm{dd}}$ are $(3-2\sqrt{2}, -1+\sqrt{2})$ and $(1,1)$, thus it follows from Corollary \ref{cor:VapsEmbed} that the volume can be computed as$\displaystyle \int^1_{\sqrt{2}-1} \int^1_{y^2} 1 \ dx dy + 2(\sqrt{2}-1)^2= \frac{7-4\sqrt{2}}{3}$
\end{enumerate}
\end{proof}

\begin{table}
 \centering
 	\begin{footnotesize}
 		\begin{tabular}{|c|cccc|}
  		\hline
  		& & & & \\
 			& $\Delta$ & $\Delta_{\rm{DLC}}$ & $\Delta_+$ & $\Delta_{\rm{dd}}$ \\
		  & & & & \\
		  \hline
		  & & & & \\
			$\frac{V(\cdot)}{ V(\Delta)}$ & 1 & $\frac{5}{12}\approx 0.4166666667$ & 0.375 & 0.25 \\
		  & & & & \\
			$\frac{V(\Delta_{\rm{emb}} \cap \ \cdot ) }{V(\cdot)}$ & $\frac{1+e^{-3\pi}}{3} \approx$ 0.3477379727 & 0.8 & $\frac{8}{9}\approx 0.8888888889$ & $\frac{14-8\sqrt{2}}{3} \approx$ 0.895430501 \\
		  & & & & \\
			$\frac{V(\Delta_{\rm{idf}} \cap \ \cdot ) }{V(\Delta_{\rm{emb}} \cap \ \cdot )}$ & $\frac{1-e^{-6\pi}}{1+e^{-3\pi}}\approx$ 0.9999193000 & $1-e^{-6\pi} \approx$ 0.9999999935 & $1-e^{-6\pi} \approx$ 0.9999999935 & 1 \\
		  & & & & \\
		  \hline
		\end{tabular}
 \end{footnotesize}
 \caption{\label{tab:volume} Relative volumes {of the spaces labelling the columns of the table (which are referred to as ``$\,\cdot\,$'' in each row)} within the set of K80 Markov matrices (first row), relative volumes of embeddable matrices within those spaces (second row) and relative volume of those of them with identifiable rates (third row). The values are rounded to the 10th decimal.}
\end{table}
The table \ref{tab:volume} shows some relative volumes of the regions defined at the beginning of this section, and is included here for quick reference. The computations involved are straightforward from the values of the volumes of the preceding proposition.

\section{Discussion} \label{sec:Discussion}

We have studied the embeddability and the identifiability of mutation rates for the Markov matrices in the K80 model of nucleotide substitution. With the results of the present paper and \citep{JJ}, the problem of embeddability is completely solved for the K81 model and its submodels K80 and JC69. One of the relevant results of the present paper is proving the existence of an open subset of embeddable K80 matrices whose principal logarithm is not a rate matrix {(see Remark \ref{rmk_Log})}.
Moreover, we have also  provided an open set of  matrices in the model that have infinitely many Markov generators (Theorem 1.1 $b.iii)$ and $c.ii)$). 
Note that this might lead to confusing results when trying to infer parameters as it is usually done in phylogenetics via a maximum likelihood approach. We further develop this issue in what follows.


In phylogenetic trees that evolve under a Markov evolutionary model one usually assigns a branch length to the edges of the tree (see Remark \ref{rk:BranchLength}). This length accounts for the expected number of substitutions per site that have occurred along that edge. It is well known that, if a K80 Markov matrix $M$ has governed evolution on an edge, then the expected number of elapsed substitutions per site can be approximated by $l(M)=\frac{-1}{4}\log(\det(M))$ {as long as $M$ is a product of Markov matrices close to the identity \citep{barry_hartigan}. According to Corollary \ref{cor:brlength}, there are embeddable K80 matrices $M$ with $l(M)=\pi$ whose rates are not identifiable and whose principal logarithm is not a rate matrix (Example \ref{ex:negative}). Therefore, there are K80 matrices with branch length $\pi$ whose embeddability property cannot be decided by looking at the principal logarithm (this is a usual practice when deciding embeddability, see  \citet{verbyla}, and might lead to erroneous results as shown here). The expected number of substitutions per site is related to molecular rate: different species evolve at different molecular rates and therefore, whether an expected number of substitutions equal to $\pi$ is large or not, depends on the species under consideration. This notion can be also related to (astronomical) time if a molecular clock can be assumed. As an example, the molecular rate is estimated to be in the range of 0.1-10 substitution per site per million year for bacteria \citep{Duchene}, 0.03-0.1 for birds or 0.2-0.8 for humans \citep{Ho2007}. Therefore, a branch of length $\pi$ could represent from $\pi/10$  to $\pi/0.03\approx 104.72$ million years, depending on the species considered (at least for the examples provided above).}

{On the other hand, as noted in Corollary \ref{cor:CharOfEmb}, these cases where embeddability cannot be decided by looking at the principal logarithm correspond to negative multiple eigenvalue $y$ (Remark \ref{rmk_Log}). Although these cases might not seem biologically realistic, there are embeddable K80 matrices with negative eigenvalues that are closer to the identity matrix than some other embeddable K80 matrices with positive eigenvalues. For more general models, one might be able to find examples of embeddable matrices even closer to the identity matrix with such a behaviour (work in progress).}

We have obtained that K80 DLC matrices are embeddable if and only if its principal logarithm is a rate matrix (Remark \ref{DLC}). In parameter estimation in phylogenetics, {it might be relevant to restrict to a subset of matrices where identifiability of the substitution parameters is guaranteed (for a discussion see \citealp{Zou2011} and \citealp{Kaehler2015}). The set of all DLC matrices is one of these subsets} \citep[see][]{Chang}.
We have found an open subset of DLC matrices whose rates are not identifiable (see Corollary \ref{DLC}), which actually have infinitely many Markov generators. This implies that different continuous time processes could lead to the same final observations, so one has to be careful when doing maximum likelihood estimation. The DLC matrix that is closest to the identity and has non-identifiable rates corresponds to branch length $2\pi$, see Example \ref{ex:pos_eig} and Theorem \ref{thm:vapsIdentiff}.  

In a work in progress, we are studying a more general Markov model which allows complex eigenvalues and therefore opens the door to exploring ill-behavior embeddable matrices (in terms of principal log not being real) that are closer to the identity matrix. Moreover, initial studies for the general Markov models suggest the existence of an open subset of embeddable Markov matrices with non-identifiable rates. We will address this case in a forthcoming paper. The identifiability of rates for JC69 matrices is a problem not solved in the present paper (the case of an eigenvalue with multiplicity 3 adds some technical difficulties), but we expect to solve it in future work.

The computation of relative volumes carried out in Section 5 shows that less than 35\% of the Markov matrices within the K80 model are embeddable. This suggests that restricting to the continuous-time approach within this model supposes a strong restriction for the inference of the Markov matrix that rules a given evolutionary process. {At the same time, this restriction guarantees a biological realism in the modeling process (see Theorem 2.2 by \citet{JJ}) and might be an advantage for reconstruction methods that estimate the parameters of the model. On the other hand, restricting to continuous-time models might lead to overestimation of genetic distances for nonstationary data \citep{Kaehler2015} and, in general, the choice of an accurate model for the data is indispensable.} The computations in Section 5 also exhibit that even if K80 embeddable matrices with non-identifiable rates describe a subset of positive measure, this case is quite marginal and represents less than 1 over 10.000 embeddable matrices.

\section*{Author contribution}

MC and JFS conceived the project, revised the proofs and computations and drafted part of the manuscript. JRL wrote the core of the manuscript and worked out the proofs and computations. All authors read, revised and approved the final manuscript.

\section*{Acknowledgements}
All authors are partially funded by AGAUR Project 2017 SGR-932 and MINECO/FEDER Projects MTM2015-69135 and MDM-2014-0445. J Roca-Lacostena has received also funding from Secretaria d'Universitats i Recerca de la Generalitat de Catalunya (AGAUR 2018FI\_B\_00947) and European Social Funds. The authors  would like to express their gratitude to Jeremy Sumner for his remarks and interesting conversations on the topic. 

\bibliographystyle{spbasic}

\end{document}